\documentclass{bmathcryptology}

\volume{1}
\issue{1}

\startpage{1}

\receiveddate{{\today}}
\revisiondate{{\today}}
\accepteddate{{\today}}

\articletitle{Higher-degree supersingular group actions}

\articleauthors{Mathilde Chenu\aff{1,2} and Benjamin Smith\aff{2,1}}

\articleaffiliations{
	 \aff{1}Laboratoire d'Informatique (LIX), CNRS, École polytechnique, Institut Polytechnique de Paris, Palaiseau, France
     \\
	 \aff{2}Inria, Palaiseau, France
}	 

\correspondingauthoremail{mathilde.chenu@inria.fr ; benjamin.smith@inria.fr}

\citationauthors{Chenu, M. \& Smith, B.} % Shorthand author list for the headers

\keywords{Isogeny-based cryptography, Supersingular elliptic curves, Endomorphisms} % Keywords separated by commas

\MSClassification{94A60, 11Y40} % 2010 Mathematics Subject Classifications separated by commas

\abstract{
    We investigate the isogeny graphs of supersingular elliptic curves over
\(\FF_{p^2}\) equipped with a \(d\)-isogeny to their Galois conjugate.
These curves are interesting because they are, in a sense,
a generalization of curves defined over \(\FF_p\),
and there is an action of the ideal class group of \(\QQ(\sqrt{-dp})\) on the isogeny graphs.
We investigate constructive and destructive aspects of these graphs in isogeny-based cryptography,
including generalizations of the CSIDH cryptosystem and the Delfs--Galbraith algorithm.

}

\usepackage[utf8]{inputenc}
\usepackage{amsmath,amsfonts,amssymb}
\usepackage{amsthm}
\usepackage{mathtools}
\usepackage{multirow}
\usepackage{xcolor}
\usepackage{xspace}
\usepackage[ruled,vlined,linesnumbered]{algorithm2e}
\DontPrintSemicolon
\usepackage{tikz}
\usetikzlibrary{cd}
\usepackage{xargs}

\newcommand{\ZZ}{\ensuremath{\mathbb{Z}}\xspace}
\newcommand{\QQ}{\ensuremath{\mathbb{Q}}\xspace}

\newcommand{\FF}{\ensuremath{\mathbb{F}}\xspace}
\newcommand{\FFbar}{\ensuremath{\overline{\mathbb{F}}}\xspace}
\newcommand{\CC}{\ensuremath{\mathbb{C}}\xspace}

\newcommand{\End}{\ensuremath{\operatorname{End}}}

\newcommand{\EC}{\ensuremath{\mathcal{E}}\xspace}

\newcommand{\pconj}[1]{{#1}^{(p)}}
\renewcommand{\neg}[1]{-{#1}}
\newcommand{\dualof}[1]{\widehat{#1}}
\newcommand{\subgrp}[1]{\ensuremath{\left\langle{#1}\right\rangle}}

\newcommand{\OO}{\mathcal{O}}
\newcommand{\OK}{\ensuremath{\OO_K}}

\renewcommand{\SS}[2][p]{\ensuremath{\mathit{SS}_{#2}(#1)}}
\newcommand{\SSpr}[2][p]{\ensuremath{\mathit{SS}_{#2}^{\text{pr}}(#1)}}

\newcommand{\Cl}{\operatorname{Cl}}

\newcommandx{\Dd}[3][1=d,2={\epsilon},3=p]{\ensuremath{\mathcal{D}_{#1,#2}(#3)}}
\newcommandx{\Ddmax}[3][1=d,2={\epsilon},3=p]{\ensuremath{\mathcal{D}_{#1,#2}^\mathrm{max}(#3)}}
\newcommandx{\Ddsub}[3][1=d,2={\epsilon},3=p]{\ensuremath{\mathcal{D}_{#1,#2}^\mathrm{sub}(#3)}}

\newcommand{\Graph}[1]{\ensuremath{\Gamma(#1)}}
\newcommand{\ellGraph}[2][\ell]{\ensuremath{\Gamma_{#1}(#2)}}

\newcommand{\Ell}{\ensuremath{\mathit{Ell}}}  % <-- ugly but let's keep it for now

\newcommand{\fraka}{\ensuremath{\mathfrak{a}}\xspace}
\newcommand{\frakb}{\ensuremath{\mathfrak{b}}\xspace}
\newcommand{\frakd}{\ensuremath{\mathfrak{d}}\xspace}
\newcommand{\frakell}{\ensuremath{\mathfrak{l}}\xspace}
\newcommand{\frakp}{\ensuremath{\mathfrak{p}}\xspace}
\newcommand{\frakt}{\ensuremath{\mathfrak{t}}\xspace}

\begin{document}

\begin{NoHyper}
\articleinformation % Output the article information section populated from the article information commands above
\end{NoHyper}

%----------------------------------------------------------------------------------------
%	ARTICLE BODY
%----------------------------------------------------------------------------------------

\section{%%%%%%%%%%%%%%%%%%%%%%%%%%%%%%%%%%%%%%%%%%%%%%%%%%%%%%%%%%%%%%%%%%%%%%%
    Introduction
}%%%%%%%%%%%%%%%%%%%%%%%%%%%%%%%%%%%%%%%%%%%%%%%%%%%%%%%%%%%%%%%%%%%%%%%%%%%%%%%
\label{sec:intro}

Supersingular isogeny graphs of elliptic curves over \(\FF_{p^2}\),
with their rich number-theoretic and combinatorial properties,
are at the heart of an increasing number of 
pre- and post-quantum cryptosystems.
Identifying and exploiting special subgraphs
of supersingular isogeny graphs
is key to understanding their mathematical properties,
and their cryptographic potential.

Isogeny-based cryptosystems
roughly fall into two families.
On one hand,
we have the cryptosystems that work in the full \(\ell\)-isogeny graph
for various \(\ell\),
including
the Charles--Goren--Lauter hash~\cite{2009/Charles--Goren--Lauter},
SIDH~\cite{2011/Jao--De-Feo,2014/De-Feo--Jao--Plut,2016/Costello--Longa--Naehrig},
SIKE~\cite{SIKE},
OSIDH~\cite{2020/Colo--Kohel},
SQISign~\cite{2020/SQISign},
and many more.
These systems take advantage of the fact that the supersingular
\(\ell\)-isogeny graph is a large regular graph with large diameter and
excellent expansion and mixing properties (indeed, it is a Ramanujan graph).

On the other hand,
we have cryptosystems that work in the \(\FF_p\)-subgraph supported on
vertices defined over \(\FF_{p}\) (or with \(j\)-invariants in \(\FF_p\)),
such as CSIDH~\cite{2018/Castryck--Lange--Martindale--Panny--Renes},
CSI-FiSh~\cite{2019/Beullens--Kleinjung--Vercauteren},
and CSURF~\cite{2020/Castryck--Decru}.
These cryptosystems, many of which represent optimizations and extensions
of pioneering work with ordinary curves
due to Stolbunov~\cite{2006/Rostovtsev--Stolbunov,2010/Stolbunov}
and Couveignes~\cite{2006/Couveignes},
take advantage of the fact that the \(\FF_p\)-endomorphism rings of these curves
are an imaginary quadratic ring,
and the ideal class group of this ring
has a convenient and efficiently-computable commutative action 
on the \(\FF_p\)-subgraph.
This group action allows us to define many simple and useful
cryptosystems,
but it also explains the structure of the \(\FF_p\)-subgraph,
allowing us to use it as a cryptanalytic tool~\cite{2016/Delfs--Galbraith}
and as a convenient point-of-reference when exploring structures in the full
supersingular isogeny graph~\cite{2021/Arpin--Camacho-Navarro--Lauter--Lim--Nelson--Scholl--Sotakova}.

This paper investigates a family of generalizations of the
\(\FF_p\)-subgraph,
one for each squarefree integer \(d\).
The key is to recognise that a curve \(\EC/\FF_{p^2}\)
has its \(j\)-invariant in \(\FF_p\) precisely when 
\(\EC\) is isomorphic to the conjugate curve
\(\pconj{\EC}/\FF_{p^2}\) 
defined by \(p\)-th powering the coefficients of \(\EC\),
and the edges of the \(\FF_p\)-subgraph
correspond to isogenies that are compatible with these isomorphisms.
In this paper,
we relax the isomorphisms to \(d\)-isogenies,
and consider the cryptographic consequences.
We obtain a series of distinguished subgraphs of the supersingular
isogeny graph,
each equipped with a free and transitive action by an ideal class group.

We define \emph{\((d,\epsilon)\)-structures}---essentially,
curves with a \(d\)-isogeny to their conjugate---and the 
isogenies between them in~\S\ref{sec:structures}.
While \((d,\epsilon)\)-structures are defined over~\(\FF_{p^2}\), 
in~\S\ref{sec:modular}
we show that they have modular invariants in \(\FF_p\),
and give useful parameterizations for \(d = 2\) and \(3\).
We narrow our focus to supersingular curves in~\S\ref{sec:supersingular},
using the theory of orientations to set up the class group action
on \((d,\epsilon)\)-structures.
We give some illustrative examples of isogeny graphs of
supersingular \((d,\epsilon)\)-structures in~\S\ref{sec:ssgraph},
before turning to cryptographic applications in~\S\ref{sec:crypto}.

Isogeny graphs of \((d,\epsilon)\)-structures
are a natural setting for
variants of CSIDH (and closely related cryptosystems).
We give arguments for the security of such cryptosystems
in~\S\ref{sec:hard-problems}.
We outline a non-interactive key exchange in~\S\ref{sec:NIKE},
generalizing CSIDH (which is the special case \(d = 1\)),
and highlight some of the subtleties that appear when we move to \(d > 1\).
Optimized implementation techniques are beyond the scope of this article.

The isogeny graphs formed by \((d,\epsilon)\)-structures form
interesting geographical features in the full supersingular isogeny
graph.
Charles, Goren, and Lauter
investigated random walks that happen to hit \((d,\pm1)\)-structures
in the security analysis of their hash
function~\cite[\S7]{2009/Charles--Goren--Lauter};
random walks into \((\ell,\pm1)\)-structures
are also key in the path-finding algorithm
of~\cite{2020/Eisentraeger--Hallgren--Leonardi--Morrison--Park}.
Further heuristics in this direction appear
in~\cite{2021/Arpin--Camacho-Navarro--Lauter--Lim--Nelson--Scholl--Sotakova}.
Here, we consider these vertices not in isolation,
but within their own isogeny graphs;
thus, we obtain a series of generalizations of the ``spine''
of~\cite{2021/Arpin--Camacho-Navarro--Lauter--Lim--Nelson--Scholl--Sotakova},
and a broad generalization of the Delfs--Galbraith isogeny-finding
algorithm~\cite{2016/Delfs--Galbraith} in~\S\ref{sec:Delfs--Galbraith}.

\paragraph{Notation and conventions.}
If \(\EC\) is an elliptic curve,
then \(\End(\EC)\) denotes its endomorphism ring
and \(\End^0(\EC)\) denotes \(\End(\EC)\otimes\QQ\).
Each elliptic curve \(\EC/\FF_{p^2}\)
has a Galois-conjugate curve \(\pconj{\EC}\),
defined by \(p\)-th powering all of the coefficients in the defining
equation of \(\EC\).
The curve and its conjugate are connected
by inseparable ``Frobenius'' \(p\)-isogenies
\(\pi_p: \EC \to \pconj{\EC}\)
and \(\pi_p: \pconj{\EC} \to \EC\),
defined by \(p\)-th powering the coordinates
(abusing notation, all inseparable \(p\)-isogenies
will be denoted by \(\pi_p\)).
Observe that \(\pconj{(\pconj{\EC})} = \EC\),
and the composition of \(\pi_p: \EC \to \pconj{\EC}\)
and \(\pi_p: \pconj{\EC} \to \EC\)
is the \(p^2\)-power Frobenius endomorphism
\(\pi_{\EC}\) of \(\EC\).
Conjugation also operates on isogenies: 
each isogeny \(\phi: \EC \to \EC'\)
defined over \(\FF_{p^2}\)
has a Galois conjugate isogeny
\(\pconj{\phi}: \pconj{\EC} \to \pconj{\EC'}\),
defined by \(p\)-th powering all of the
coefficients
in a rational map defining \(\phi\).
We always have
\[
    \pconj{(\pconj{\phi})} = \phi
    \qquad
    \text{and}
    \qquad 
    \pi_p\circ\phi = \pconj{\phi}\circ\pi_p
    \,,
\]
and
conjugation thus gives an isomorphism
of rings between \(\End(\EC)\) and \(\End(\pconj{\EC})\).

\section{%%%%%%%%%%%%%%%%%%%%%%%%%%%%%%%%%%%%%%%%%%%%%%%%%%%%%%%%%%%%%%%%%%%%%%%
    Curves with a \texorpdfstring{\protect{\lowercase{\(d\)}}}{d}-isogeny to their conjugate
}%%%%%%%%%%%%%%%%%%%%%%%%%%%%%%%%%%%%%%%%%%%%%%%%%%%%%%%%%%%%%%%%%%%%%%%%%%%%%%%
\label{sec:structures}

Let \(p > 3\) be a prime,
and \(d\) a squarefree integer prime to \(p\).
Typically, \(p\) is very large
and \(d\) is very small.

We are interested in elliptic curves \(\EC/\FF_{p^2}\)
equipped with a \(d\)-isogeny \(\psi: \EC \to \pconj{\EC}\).
Given any such \(d\)-isogeny \(\psi\),
we have two returning \(d\)-isogenies:
\[
    \pconj{\psi}: \pconj{\EC} \to \EC
    \qquad
    \text{and}
    \qquad
    \dualof{\psi}: \pconj{\EC} \to \EC
    \,.
\]

\begin{definition}
    Let \(\EC/\FF_{p^2}\) be an elliptic curve
    equipped with a \(d\)-isogeny
    \(\psi: \EC \to \pconj{\EC}\)
    to its conjugate.
    We say that \((\EC,\psi)\) is a \emph{$(d,\epsilon)$-structure}
    if
    \[
        \dualof{\psi} = \epsilon \pconj{\psi}
        \quad
        \text{with}
        \quad
        \epsilon \in \{1,-1\}
        \,.
    \]
    Each \((d,\epsilon)\)-structure
    \((\EC,\psi)\)
    has an \emph{associated endomorphism}
    \[
        \mu := \pi_p\circ\psi \in \End(\EC)
        \,.
    \]
    We say that \((\EC,\psi)\) is ordinary resp. supersingular
    if \(\EC\) is ordinary resp. supersingular.\footnote{%
        We focus on curves defined over \(\FF_{p^2}\) 
        because our applications involve supersingular curves,
        and every supersingular curve is isomorphic to a curve over~\(\FF_{p^2}\).
        One might consider isogenies to conjugates over higher-degree
        extensions, but then in general we do not have 
        the relation \(\dualof{\psi} = \pm\pconj{\psi}\),
        which is fundamental to our results.
    }

\end{definition}

\begin{proposition}
    \label{prop:r}
    If \((\EC,\psi)\) is a \((d,\epsilon)\)-structure
    and \(\mu\) is its associated endomorphism,
    then
    \[
        \mu^2 = [\epsilon d]\pi_{\EC}
        \,.
    \]
    If \(\pi_\EC\) is the Frobenius endomorphism of~\(\EC\)
    and \(t_\EC\) is its trace,
    then there exists an integer \(r\) such that
    \(
        [r]\mu = [p] + \epsilon \pi_\EC
    \)
    in \(\End(\EC)\),
    \(
        dr^2 = 2p + \epsilon t_\EC
    \)
    in \(\ZZ\),
    and the characteristic polynomial of \(\mu\) is
    \(
        P_\mu(T) = T^2 - rdT + dp
    \).
\end{proposition}
\begin{proof}
    We have \(\psi\pi_p = \pi_p\pconj{\psi}\),
    so
    \(
        \mu^2 
        = 
        \pi_p\psi\pi_p\psi
        =
        \pi_p(\pi_p\pconj{\psi})\psi
        = 
        \pi_{\EC}(\pconj{\psi}\psi)
    \).
    Now \(\pconj{\psi} = \epsilon\dualof{\psi}\)
    (because \((\EC,\psi)\) is a \((d,\epsilon)\)-structure),
    so \(\pconj{\psi}\psi = [\epsilon d]\),
    and therefore
    \(\mu^2 = [\epsilon d]\pi_{\EC}\).
    For the rest: $\mu$ has degree $dp$,
    so it satisfies a quadratic polynomial
    $P_\mu(T) = T^2 - aT + dp$ for some integer $a$. 
    The first assertion then implies
    $[a]\mu = \mu^2 + [dp] = [\epsilon d] \pi_\EC + [dp]$.
    Squaring, we obtain
    	$$([a]\mu)^2 
    	 = [d]^2(\pi_\EC^2 +p^2) + 2[dp][\epsilon d] \pi_\EC
    	 = [d]^2(t_\EC\pi_\EC) + 2[dp][\epsilon d] \pi_\EC
    	 = [\epsilon d] \pi_\EC  ([\epsilon d]  t_\EC + 2dp)\,, $$
  	so $a^2 = \epsilon d  t_\EC + 2dp$, hence $d \mid a^2$.
    But $d$ is squarefree, so \(d \mid a\),
    and then $r = a/d$ satisfies the given conditions.
\end{proof}

\begin{remark}
    In the situation of Proposition~\ref{prop:r}:
    if \(\EC\) is ordinary,
    then $\ZZ[\mu]$ and $\ZZ[\pi_{\EC}]$ are
    orders in \(\QQ(\pi_{\EC})\)
    of discriminant $d^2r^2 - 4dp$ and 
  	$t_\EC^2 - 4p^2=r^2(d^2r^2 - 4dp)$, respectively,
    so
    \(|r|\) is the conductor of \(\ZZ[\pi_{\EC}]\) in~\(\ZZ[\mu]\).
    (The supersingular case is treated in detail
    in~\S\ref{sec:supersingular}.)
\end{remark}

\begin{definition}
    \label{def:isogeny-of-structures}
    \label{def:isomorphism-of-structures}
    Let \((\EC,\psi)\) and \((\EC',\psi')\) be \((d,\epsilon)\)-structures.
    We say an isogeny (resp.~isomorphism) \(\phi: \EC \to \EC'\)
    is an \emph{isogeny (resp.~isomorphism) of \((d,\epsilon)\)-structures}
    if 
    \(
        \psi'\phi = \pconj{\phi}\psi
    \),
    that is, if the following diagram commutes:
    \[
        \begin{tikzcd}[sep=huge]
            \EC
            \arrow[r, "\psi"]
            \arrow[d, "\phi"]
            & 
            \pconj{\EC}
            \arrow[d, "\pconj{\phi}"]
            \\
            \EC'
            \arrow[r, "\psi'"]
            &
            \pconj{(\EC')}
        \end{tikzcd}
    \]
\end{definition}

It is easily verified that isogenies of \((d,\epsilon)\)-structures
follow the usual rules obeyed by isogenies:
the composition of two isogenies of \((d,\epsilon)\)-structures
is an isogeny of \((d,\epsilon)\)-structures,
the dual of an isogeny of \((d,\epsilon)\)-structures
is an isogeny of \((d,\epsilon)\)-structures,
and every \((d,\epsilon)\)-structure
has an isogeny to itself (the identity map, for example).
Isogeny therefore forms an equivalence relation on
\((d,\epsilon)\)-structures.

If \((\EC, \psi)\) is a \((d,\epsilon)\)-structure
with associated endomorphism \(\mu\),
then 
\[
    \neg{(\EC,\psi)} := (\EC,-\psi)
    \qquad
    \text{and}
    \qquad
    \pconj{(\EC,\psi)} := (\pconj{\EC},\pconj{\psi})
\]
are \((d,\epsilon)\)-structures
with associated endomorphisms \(-\mu\) 
and \(\pconj{\mu}\),
respectively.
If \(\phi: (\EC,\psi) \to (\EC',\psi')\)
is an isogeny of \((d,\epsilon)\)-structures,
then \(\phi: \neg{(\EC,\psi)} \to \neg{(\EC',\psi')}\)
and \(\pconj{\phi}: \pconj{(\EC,\psi)} \to \pconj{(\EC',\psi')}\)
are also isogenies of \((d,\epsilon)\)-structures.
We thus have two involutions,
\textbf{negation} and \textbf{conjugation},
on the category of \((d,\epsilon)\)-structures and their isogenies.

\begin{remark}
    The isogenies \(\psi\)
    and \(\pi_p: \EC \to \pconj{\EC}\)
    are both in fact isogenies of \((d,\epsilon)\)-structures
    \((\EC,\psi) \to \pconj{(\EC,\psi)}\).
\end{remark}

\paragraph{Twisting.}
Let \(\alpha\) be an element of \(\FFbar_p\setminus\{0\}\).
For each elliptic curve \(\EC: y^2 = x^3 + ax + b\),
there is a curve
\[
    \EC^\alpha/\FF_{p^2}(\alpha^2): y^2 = x^3 + \alpha^4 ax + \alpha^6 b
\]
and an \(\FF_{p^2}(\alpha)\)-isomorphism
\(\tau_\alpha: \EC \to \EC^\alpha\)
defined by \((x,y) \mapsto (\alpha^2 x, \alpha^3y)\).
Abusing notation,
we write \(\tau_\alpha\) for this map
on \emph{every} elliptic curve;
with this convention,
\(\tau_\beta\circ\tau_\alpha = \tau_{\alpha\beta}\).
If \(\delta\) is a nonsquare in \(\FF_{p^2}\)
then \(\EC^{\sqrt{\delta}}\) is the \emph{quadratic twist}
(which, up to \(\FF_{p^2}\)-isomorphism,
is independent of the choice of nonsquare \(\delta\))
and \(\tau_{\sqrt{\delta}}\) is the twisting isomorphism.
For each isogeny \(\phi: \EC \to \EC'\)
defined over \(\FF_{p^2}\),
there is an \(\FF_{p^2}(\alpha^2)\)-isogeny
\[
    \phi^\alpha
    := 
    (\tau_\alpha\circ\phi\circ\tau_{1/\alpha})
    : 
    \EC^\alpha 
    \longrightarrow 
    (\EC')^\alpha
    \,.
\]

Now let \((\EC,\psi)\) be a \((d,\epsilon)\)-structure
with associated endomorphism \(\mu\).
If again we choose a nonsquare \(\delta\) in \(\FF_{p^2}\),
and a square root \(\sqrt{\delta}\) of \(\delta\) in~\(\FF_{p^4}\),
then in general \((\EC^{\sqrt{\delta}},\psi^{\sqrt{\delta}})\)
is \emph{not} a \((d,\pm1)\)-structure
(because conjugation and twisting generally do not commute);
but
\(
    (\EC,\psi)^{\sqrt{\delta}}
    :=
    (
        \EC^{\sqrt{\delta}},
        \tau_{(\sqrt{\delta})^{(p-1)}}\circ\psi^{\sqrt{\delta}}
    )
\)
is a \((d,-\epsilon)\)-structure
with associated endomorphism \(\mu^{\sqrt{\delta}}\).
The \(\FF_{p^2}\)-isomorphism class of \((\EC,\psi)^{\sqrt{\delta}}\)
is independent of the choice of \(\delta\);
we call \((\EC,\psi)^{\sqrt{\delta}}\)
the \emph{quadratic twist} of \((\EC,\psi)\).
Note that
\(((\EC,\psi)^{\sqrt{\delta}})^{\sqrt{\delta}} \cong (\EC,\psi)\).
If \(\phi: (\EC,\psi) \to (\EC',\psi')\)
is an isogeny of \((d,\epsilon)\)-structures,
then
\(\phi^{\sqrt{\delta}}\) induces an isogeny
of \((d,-\epsilon)\)-structures
\(
    \phi^{\sqrt{\delta}}: 
    (\EC,\psi)^{\sqrt{\delta}}
    \to
    (\EC',\psi')^{\sqrt{\delta}}
\).
Twisting therefore takes us
from the category of \((d,\epsilon)\)-structures 
into the category of \((d,-\epsilon)\)-structures
and back again.

\begin{example}
    Consider the case \(d = 1\).
    Each \((1,1)\)-structure is
    \(\FF_{p^2}\)-isomorphic to the base-extension to \(\FF_{p^2}\)
    of a curve defined over \(\FF_{p}\)
    (with the \(1\)-isogeny being \([\pm1]\));
    the associated endomorphism is the \(p\)-power Frobenius
    endomorphism on the base-extended curve,
    and the integer \(r\) of Proposition~\ref{prop:r}
    is the trace of the \(p\)-power Frobenius.
    Each \((1,-1)\)-structure is the quadratic twist of a
    \((1,1)\)-structure:
    essentially, an ordinary \((1,-1)\)-structure is isomorphic to a GLS
    curve~\cite{2011/Galbraith--Lin--Scott}.
    This discussion should be compared with the remark at the end
    of~\cite[\S3]{2016/Smith}.
\end{example}

\section{%%%%%%%%%%%%%%%%%%%%%%%%%%%%%%%%%%%%%%%%%%%%%%%%%%%%%%%%%%%%%%%%%%%%%%%
    Parametrizations and modular curves
}%%%%%%%%%%%%%%%%%%%%%%%%%%%%%%%%%%%%%%%%%%%%%%%%%%%%%%%%%%%%%%%%%%%%%%%%%%%%%%%
\label{sec:modular}

For our computations,
we can represent a \((d,\epsilon)\)-structure \((\EC,\psi)\)
as \((\EC,f_\psi,\alpha)\),
where \(f_\psi\) is the kernel polynomial of \(\psi\)
(that is, the monic polynomial whose roots are the \(x\)-coordinates of
the nonzero points in \(\ker\psi\))
and \(\alpha\) is the element such that 
\(\psi = \tau_\alpha\circ\tilde{\psi}\),
where \(\tilde{\psi}: \EC \to \EC/\ker\psi\)
is the normalized ``Vélu'' isogeny.

We want a more space-efficient encoding of
isomorphism classes of \((d,\epsilon)\)-structures,
both as a canonical encoding for vertices in isogeny graphs,
and for transmission of \((d,\epsilon)\)-structures
used as cryptographic values.

While \((d,\epsilon)\)-structures
may seem to be relatively complicated objects over \(\FF_{p^2}\),
their isomorphism classes can be encoded to little more than a single
element of \(\FF_p\).
Briefly: the key is to take the quotient by negation,
which maps the set \(S_{d,\epsilon}\) of isomorphism classes of \((d,\epsilon)\)-structures over \(\FF_{p^2}\)
into \(X_0(d)(\FF_{p^2})\),
where \(X_0(d)\) is the level-\(d\) modular curve.
Then, the Atkin--Lehner involution \(\omega_d\),
which maps a modular point onto its ``dual'',
acts as conjugation on the image of \(S_{d,\epsilon}\).
Writing \(X_0^+(d) = X_0(d)/\subgrp{\omega_d}\),
we have a four-to-one map from \(S_{d,\epsilon}\)
onto \(X_0^+(d)(\FF_p)\),
identifying the isomorphism class of \((\EC,\psi)\)
with \(-(\EC,\psi)\),
\(\pconj{(\EC,\psi)}\),
and \(-\pconj{(\EC,\psi)}\).
We can therefore represent an element of \(S_{d,\epsilon}\)
as a point in \(X_0^+(d)(\FF_p)\)
plus two bits (one to determine the sign, the other the conjugate).
Since \(X_0^+(d)\) is a curve,
we can further compress the representative point in \(X_0^+(d)(\FF_p)\)
to one element of \(\FF_p\)
plus a few bits.
This step depends strongly on the geometry of \(X_0^+(d)(\FF_p)\):
for example, if \(X_0^+(d)\) has genus \(0\)
then we can rationally parametrize it, giving a simple compression of points in
\(X_0^+(d)(\FF_p)\) to single elements of \(\FF_p\);
if \(X_0^+(d)\) is hyperelliptic,
then we can compress points in \(X_0^+(d)(\FF_p)\) 
to a single element of \(\FF_p\) plus a ``sign'' bit 
in the usual way;
and as the gonality of \(X_0^+(d)\) increases,
so does the number of auxiliary bits required.

A full development of these representations 
and the algorithms that operate on them is 
beyond the scope of this short article, 
but we will give useful explicit constructions 
for \(d = 2\) and \(3\) here,
derived from explicit parametrizations of \(\QQ\)-curves
due to Hasegawa~\cite{1997/Hasegawa}.
The associated endomorphisms 
for ordinary curves in these families
have been used to accelerate
scalar multiplication algorithms
(see~\cite{2016/Smith},
where we also find related families for \(d = 5\) and \(7\),
and~\cite{2013/Guillevic--Ionica})
and as inputs for specialized point-counting algorithms~\cite{2016/Morain--Scribot--Smith}.

\subsection{Representing \texorpdfstring{\((2,\epsilon)\)}{(2,e)}-structures}
\label{sec:Hasegawa-2}

Let \(\Delta\) be a nonsquare in \(\FF_p\),
and fix a square root \(\sqrt{\Delta}\) in \(\FF_{p^2}\).
For each \(u\) in \(\FF_p\), the curve
\[
    \EC_{2,u}/\FF_{p^2}: 
    y^2 
    = 
    x^3 - 6(5 - 3u\sqrt{\Delta})x + 8(7 - 9u\sqrt{\Delta})
\]
has a rational \(2\)-torsion point \((4,0)\),
which generates the kernel of a \(2\)-isogeny 
\(\psi_{2,u}: \EC_{2,u} \to \pconj{\EC_{2,u}}\)
defined over~\(\FF_{p^2}\).
If we use Vélu's formulae
to compute the (normalized) quotient isogeny
\(\EC_{2,u} \to \EC_{2,u}/\subgrp{(4,0)}\),
then
the isomorphism \(\EC_{2,u}/\subgrp{(4,0)} \to \pconj{\EC_{2,u}}\)
is \(\tau_{1/\sqrt{-2}}\).
Composing, we obtain an expression for \(\psi_{2,u}\) as a rational map:
\[
    \label{eq:phi2}
    \psi_{2,u}:
    (x,y)
    \longmapsto
    \left(
    \frac{-x}{2} - \frac{9(1 + u\sqrt{\Delta})}{x-4}
    , 
    \frac{y}{\sqrt{-2}}
    \left(
    \frac{-1}{2} + \frac{9(1 + u\sqrt{\Delta})}{(x-4)^2}
    \right)
    \right)
    \,.
\]
Computing the dual isogeny \(\dualof{\psi}_{2,u}\)
and comparing it with \(\pconj{\psi_{2,u}}\),
we find that
\((\EC_{2,u},\psi_{2,u})\) is 
a \((2,1)\)-structure if \( p \equiv 5, 7 \pmod{8} \),
or 
a \((2,-1)\)-structure if \( p \equiv 1, 3 \pmod{8} \).
(To obtain a family of \((2,-1)\)-structures
when \(p \equiv 5, 7 \pmod{8}\)
or \((2,1)\)-structures if \(p \equiv 1, 3 \pmod{8}\),
it suffices to take the quadratic twist.)

\subsection{Representing \texorpdfstring{\((3,\epsilon)\)}{(3,e)}-structures}
\label{sec:Hasegawa-3}

Let \(\Delta\) be a nonsquare in \(\FF_p\),
and fix a square root \(\sqrt{\Delta}\) in \(\FF_{p^2}\).
For each \(u\) in \(\FF_p\), the elliptic curve
\[
    \EC_{3,u}/\FF_{p^2}: 
    y^2 =
    x^3 - 3\big(5 + 4u\sqrt{\Delta}\big)x 
    + 2\big(2u^2\Delta + 14u\sqrt{\Delta} + 11\big)
\]
has an order-\(3\) subgroup
\(\{\OO, (3,\pm2(1-u\sqrt{\Delta}))\}\)
defined by the polynomial \(x - 3\).
Taking the quotient with Vélu's formulae
and composing with \(\tau_{1/\sqrt{-3}}\)
yields an explicit 3-isogeny 
\(\psi_{3,u}: \EC_{3,u} \to \pconj{\EC_{3,u}}\),
and we find that
\((\EC_{3,u},\psi_{3,u})\) is 
a \((3,1)\)-structure if \(p \equiv 2 \pmod{3}\),
or 
a \((3,-1)\)-structure if \(p \equiv 1 \pmod{3}\).
(To obtain a family of \((3,-1)\)-structures
when \(p \equiv 2 \pmod{3}\)
or \((3,1)\)-structures if \(p \equiv 1 \pmod{3}\),
take the quadratic twist.)

\section{%%%%%%%%%%%%%%%%%%%%%%%%%%%%%%%%%%%%%%%%%%%%%%%%%%%%%%%%%%%%%%%%%%%%%%%
    Supersingular \texorpdfstring{\lowercase{\((d,\epsilon)\)}}{(d,e)}-structures
}%%%%%%%%%%%%%%%%%%%%%%%%%%%%%%%%%%%%%%%%%%%%%%%%%%%%%%%%%%%%%%%%%%%%%%%%%%%%%%%
\label{sec:supersingular}

We now come to the main focus of our investigation:
supersingular \((d,\epsilon)\)-structures and their isogeny graphs.

\begin{definition}
    We write $\Dd$ for the set of
    supersingular \((d,\epsilon)\)-structures over \(\FF_{p^2}\) up to
    \(\FF_{p^2}\)-isomorphism,
    and \(\Graph{\Dd}\)
    for the graph on \(\Dd\)
    whose edges are (\(\FF_{p^2}\)-isomorphism classes of)
    isogenies of \((d,\epsilon)\)-structures.
    For each prime \(\ell \not= p\),
    we write \(\ellGraph{\Dd}\)
    for the subgraph of \(\Graph{\Dd}\) where the edges are \(\ell\)-isogenies.
\end{definition}

Observe that the quadratic twist gives an isomorphism of graphs \(\Graph{\Dd} \cong \Graph{\Dd[d][-\epsilon]}\).

\begin{proposition}
    \label{prop:supersingular}
    Let \((\EC,\psi)\) be a \((d,\epsilon)\)-structure
    with associated endomorphism \(\mu\).
    If \(\EC\) is supersingular, then
    \begin{enumerate}
        \item
            \(\mu^2 = [-dp]\).
        \item
            The trace of Frobenius satisfies
            \(t_\EC = -2\epsilon p\),
            and in particular
            \(\EC(\FF_{p^2}) \cong (\ZZ/(p + \epsilon)\ZZ)^2\).
    \end{enumerate}
\end{proposition}
\begin{proof}
    With the notation of Proposition~\ref{prop:r}:
    The curve $\EC$ is supersingular if and only if $p \mid t_\EC$. 
    Now $p \nmid d$, so $p \mid r$ by Proposition~\ref{prop:r}.
    The characteristic polynomial $P_{\mu}(T)$ of $\mu$
    has discriminant $(rd)^2-4dp$;
    this discriminant cannot be positive,
    so $|r| \le 2\sqrt{p/d}$.
    Since $p \mid r$, we have $r = 0$,
    so $\mu^2 = [-dp]$, and 
    $ t_\EC = \frac{-2p}{\epsilon} = - 2\epsilon p$.
\end{proof}

Proposition~\ref{prop:supersingular}
tells us that if \((\EC,\psi)\) is a supersingular \((d,\epsilon)\)-structure,
then \(\epsilon\) is completely determined by the \(\FF_{p^2}\)-isogeny
class of \(\EC\).
Further, \(t_\EC\) can only be \(\pm 2p\):
the special supersingular traces \(-p\), \(0\), and \(p\)
(corresponding to non-quadratic twists 
of curves of \(j\)-invariant \(0\) and \(1728\),
if these are supersingular)
cannot occur.

\subsection{Orientations}%%%%%%%%%%%%%%%%%%%%%%%%%%%%%%%
\label{sec:orientations}

Proposition~\ref{prop:supersingular} tells us
that the associated endomorphism
of each supersingular \((d,\epsilon)\)-structure
acts like a square root of \(-dp\)
in the endomorphism ring.
We can make this notion more precise
using \emph{orientations},
as described by Colò and Kohel in~\cite{2020/Colo--Kohel}
and Onuki in~\cite{2021/Onuki}.
Before going further,
we recall some generalities.

Let $K$ be an imaginary quadratic field, $\OK$ its ring of integers,
and $\mathcal{O}$ an order in $K$.
A $K$-orientation on an elliptic curve $\EC/\FF_{p^2}$ is a homomorphism
$\iota : K \to \End^0(\EC)$;
we call the pair \((\EC,\iota)\) a \emph{$K$-oriented elliptic curve}.
We say $\iota$ is an $\OO$-orientation,
and \((\EC,\iota)\) is an $\OO$-oriented elliptic curve,
if $\iota(\OO) \subseteq \End(\EC)$.
An $\OO$-orientation \(\iota: K \to \End^0(\EC)\) is \emph{primitive} 
if $\iota(\OO) = \End(\EC) \cap \iota(K)$:
that is, if \(\iota\) is ``full'' in the sense that it does not extend
to an \(\OO'\)-orientation for any strict super-order \(\OO'\supset\OO\).

Let \((\EC,\iota)\) be a $K$-oriented elliptic curve.
If \(\phi: \EC \to \EC'\) is an isogeny,
then there is an \emph{induced $K$-orientation}
\(\phi_*(\iota)\) on \(\EC'\) defined by
\[
    \phi_*(\iota) : 
    \alpha
    \longmapsto
    \frac{1}{\deg(\phi)}\phi\circ\iota(\alpha)\circ\dualof{\phi}
    \,.
\]
Given two oriented curves $(\EC, \iota )$ and $(\EC', \iota' )$, 
an isogeny $\phi : \EC \to \EC'$ is said to be $K$-oriented,
or an isogeny of $K$-oriented elliptic curves,
if $\iota' = \phi_*(\iota)$.
In this case we write \(\phi: (\EC,\iota) \to (\EC',\iota')\).
If there exists a $K$-oriented isogeny
\(\tilde{\phi}: (\EC',\iota') \to (\EC,\iota)\)
such that \(\tilde{\phi}\circ\phi = [1]_{\EC}\)
and \(\phi\circ\tilde{\phi} = [1]_{\EC'}\),
then we say that \(\phi\) is a \(K\)-oriented isomorphism,
and we write \((\EC,\iota) \cong (\EC',\iota')\).
Note that 
\(\phi: (\EC,\iota) \to (\EC',\iota')\) is an oriented isomorphism
if and only if the underlying isomorphism of curves \(\phi\)
satisfies \(\phi\circ\iota(\alpha) = \iota'(\alpha)\circ\phi\)
for all \(\alpha\) in \(K\).

If \(\phi: (\EC,\iota) \to (\EC',\iota')\)
is a \(K\)-oriented isogeny,
then
\(\iota\) resp. \(\iota'\)
is a primitive \(\OO\) resp. \(\OO'\)-orientation
for some order \(\OO\) resp. \(\OO'\) in \(K\).
If \(\ell = \deg\phi\) is a prime not equal to \(p\),
then one of the following holds:
\begin{itemize}
    \item
        \(\OO = \OO'\),
        and \(\phi\) is said to be \emph{horizontal};
        or
    \item
        \(\OO \subset \OO'\) with \([\OO':\OO] = \ell\),
        and \(\phi\) is said to be \emph{ascending};
        or
    \item
        \(\OO \supset \OO'\) with \([\OO:\OO'] = \ell\),
        and \(\phi\) is said to be \emph{descending}.
\end{itemize}

Let \(\OO\) be an order in a quadratic field \(K\)
such that $p$ does not split in $K$ or divide the conductor of $\OO$. 
Following~\cite{2020/Colo--Kohel},
we let \(\SS{\OO}\) denote the set of $\OO$-oriented
supersingular elliptic curves over \(\FFbar_p\)
up to \(K\)-oriented isomorphism.
The subset of primitive $\OO$-oriented 
curves (up to \(K\)-oriented isomorphism)
is denoted by $\SSpr{\OO}$.

For any integral invertible ideal \(\fraka\) in \(\OO\)
and any \(\OO\)-oriented curve \((\EC,\iota)\),
we have a finite subgroup
\[
    \EC[\fraka] 
    := 
    \{ P \in \EC \mid \iota(\alpha)(P) = 0 \quad \forall \alpha \in \fraka \}
    \,.
\]
Now suppose \(\fraka\) is prime to the conductor of \(\OO\) in
\(\OK\).\footnote{
    Working with the class group,
    we can always replace ideals
    that are not prime to the conductor
    with equivalent integral ideals
    that are.
}
If \(\phi_\fraka: \EC \to \EC/\EC[\fraka]\) is the quotient isogeny,
then \((\phi_\fraka)_*(\iota)\) is an \(\OO\)-orientation on
\(\EC/\EC[\fraka]\),
and \(\phi_\fraka\) is a horizontal isogeny of \(\OO\)-oriented curves.
If \(\fraka\) is principal then
\(
    (\EC/\EC[\fraka],(\phi_\fraka)_*(\iota))
    \cong
    (\EC,\iota)
\),
so the map 
\[
    (\fraka,(\EC,\iota)) 
    \mapsto
    (\EC/\EC[\fraka],(\phi_\fraka)_*(\iota))
\]
extends to fractional ideals and factors through the class group,
and as in~\cite{2020/Colo--Kohel} 
we get a transitive group action
\begin{align*}
    \Cl(\OO) \times \SS{\OO} & \longrightarrow \SS{\OO}
    \,.
\end{align*}

Onuki~\cite{2021/Onuki} shows that if we restrict to a certain subset of the
\emph{primitive} \(\OO\)-oriented curves,
then this action is transitive and free.
Let \(\mathcal{J}_\OO\) denote the set of $j$-invariants of elliptic
curves $\EC$ over $\CC$ (not \(\FFbar_p\))
with $\End(\EC) \cong \OO$.
All elements in $\mathcal{J_O}$ are algebraic integers, 
so an elliptic curve whose $j$-invariant is in $\mathcal{J}_\OO$
has potential good reduction at any prime ideal. 
Since $\mathcal{J}_\OO$ is finite,
we can take a number field $L$
and a prime ideal $\frakp$ of $L$ above $p$ 
such that for all $j \in \mathcal{J_O}$,
there exists an elliptic curve over $L$ with good reduction at $\frakp$
and $j$-invariant $j$.
Fix an injection of the residue field of $L$ modulo $\frakp$ into $\FFbar_p$.
Let \(\Ell(\OO)\) be the set of isomorphism classes of elliptic curves $\EC$ over
$L$ with good reduction at \(p\) and \(j\)-invariants in $\mathcal{J_O}$.
For every such \(\EC\),
we let \([\cdot]_{\EC}\) be the \emph{normalized} \(\OO\)-orientation:
that is, such that for any invariant differential \(\omega\) on \(\EC\),
\(([\alpha]_{\EC})^*\omega = \alpha\omega\)
for all \(\alpha\) in \(\OO\).
Then reduction mod \(\frakp\) defines a map 
\(\rho : \Ell(\OO) \to  \SSpr{\OO} \)
sending \(\EC\) to \((\widetilde{\EC}, [.]_{\widetilde{\EC}})\),
where \(\widetilde{\EC}\) is the reduction of \(\EC/L\) at \(\frakp\)
and \([\cdot]_{\tilde\EC}\)
is the orientation 
such that \([\alpha]_{\widetilde{\EC}} = [\alpha]_\EC \pmod{\frakp}\)
for all \(\alpha\) in \(\OO\).

\begin{theorem}[\protect{Onuki~\cite[Theorem 3.4]{2021/Onuki}}]
	\label{theorem:Onuki}
    With the notation above:
    $\Cl(\OO)$ acts freely and transitively on $\rho(\Ell(\OO))$.
\end{theorem}

\subsection{The natural orientation}

From now on we let \(K = \QQ(\sqrt{-dp})\), and let \(\OK\) be the maximal order of~\(K\).

If \((\EC,\psi)\) is a supersingular \((d,\epsilon)\)-structure
and \(\mu\) is the associated endomorphism,
then 
\begin{align*}
    \iota_\psi: 
    \QQ(\sqrt{-dp})
    & 
    \longrightarrow
    \End^0(\EC)
    \\
    \sqrt{-dp}
    &
    \longmapsto
    \mu
\end{align*}
is a \(\ZZ[\sqrt{-dp}]\)-orientation
by Proposition~\ref{prop:supersingular}.
We call this the \emph{natural} orientation.

\begin{lemma}
    \label{lemma:orientation-to-structure}
    If \(\EC/\FF_{p^2}\) is a supersingular elliptic curve 
    with \(\#\EC(\FF_{p^2}) = (p + \epsilon)^2\)
    and \(\iota\) is a \(\ZZ[\sqrt{-dp}]\)-orientation on \(\EC\),
    then \(\iota\) is the natural orientation
    for some \((d,\epsilon)\)-structure \((\EC,\psi)\).
\end{lemma}
\begin{proof}
    Let \(\mu := \iota(\sqrt{-dp})\) in \(\End(\EC)\).
    We have \(\deg(\mu) = dp\) and \(p\nmid d\),
    so \(\mu\) factors over \(\FF_{p^2}\) 
    into the composition of a \(d\)-isogeny and a \(p\)-isogeny.
    Since \(\EC\) is supersingular,
    the \(p\)-isogeny is isomorphic to \(\pi_p\),
    and so
    \(\mu = \pi_p\psi\)
    for some \(d\)-isogeny \(\psi: \EC \to \pconj{\EC}\).
    It remains to show that \(\dualof{\psi} = \epsilon\pconj{\psi}\).
    Now \([-dp] = \mu^2 = \pi_p\psi\pi_p\psi = \pconj{\psi}\pi_p^2\psi =
    \pconj{\psi}\psi\pi_p^2\),
    and \(\pi_p^2 = [-\epsilon p]\)
    because \(\EC\) is supersingular 
    with \(\#\EC(\FF_{p^2}) = (p+\epsilon)^2\),
    so \([d] = \epsilon\pconj{\psi}\psi\),
    and therefore
    \(\dualof{\psi} = \epsilon\pconj{\psi}\).
\end{proof}

\begin{lemma}
    \label{lemma:natural-induced}
    Let \((\EC,\psi)\) and \((\EC',\psi')\) 
    be \((d,\epsilon)\)-structures
    with natural orientations \(\iota_\psi\)
    and \(\iota_{\psi'}\),
    respectively.
    If \(\phi: \EC \to \EC'\) is an isogeny,
    then
    \(\phi\) is an isogeny (resp.~isomorphism) of \(\ZZ[\sqrt{-dp}]\)-oriented
    elliptic curves \((\EC,\iota) \to (\EC',\iota')\)
    if and only if it is 
    an isogeny (resp.~isomorphism) of \((d,\epsilon)\)-structures \((\EC,\psi)\to(\EC',\psi')\).
\end{lemma}
\begin{proof}
    Let \(\mu\) resp.~\(\mu'\)
    be the associated endomorphisms 
    of \((\EC,\psi)\) resp.~\((\EC',\psi')\);
    then
    \begin{align*}
        \phi_*(\iota_{\psi}) = \iota_{\psi'}
        \iff
        \phi_*(\iota_{\psi})(\sqrt{-dp}) & = \iota_{\psi'}(\sqrt{-dp})
        & \text{(\(\sqrt{-dp}\) generates \(\QQ(\sqrt{-dp})\)}
        \\
        \iff
        \phi\circ\mu\circ\dualof{\phi} & = \mu'[\deg\phi]
        & \text{(multiplying by \(\deg\phi\))}
        \\
        \iff
        \phi\circ\mu & = \mu'\circ\phi
        & \text{(cancelling \(\dualof{\phi}\))}
        \\
        \iff
        \phi\circ\pi_p\circ\psi & = \pi_p\circ\psi'\circ\phi
        & \text{(by definition)}
        \\
        \iff
        \pi_p\circ\pconj{\phi}\circ\psi & = \pi_p\circ\psi'\circ\phi
        & \text{(\(\pi_p\circ\phi = \pconj{\phi}\circ\pi_p\))}
        \\
        \iff
        \pconj{\phi}\circ\psi & = \psi'\circ\phi
        & \text{(cancelling \(\pi_p\))}
    \end{align*}
    and the result follows on comparing definitions.
\end{proof}

Colò and Kohel~\cite{2020/Colo--Kohel} and Onuki~\cite{2021/Onuki}
use class-group actions to study the isogeny graphs \(\Graph{\SS{\OO}}\)
with vertex set \(\SS{\OO}\) for different orders \(\OO\).
Proposition~\ref{prop:graph-isomorphism}
allows us to transfer their results to our setting of \((d,\epsilon)\)-structures.

\begin{proposition}
    \label{prop:graph-isomorphism}
    The graphs \(\Graph{\Dd}\)
    and \(\Graph{\SS{\ZZ[\sqrt{-dp}]}}\)
    are explicitly isomorphic
    for \(\epsilon = 1\) and \(\epsilon = -1\).
\end{proposition}
\begin{proof}
    This follows from Lemmas~\ref{lemma:orientation-to-structure}
    and~\ref{lemma:natural-induced},
    once we can show that the isomorphism class of any
    \(\ZZ[\sqrt{-dp}]\)-oriented supersingular curve \((\EC,\iota)\) over \(\FFbar_p\)
    contains a representative over \(\FF_{p^2}\) of order \((p+\epsilon)^2\).
    Since \(j(\EC)\) is in \(\FF_{p^2}\),
    after a suitable \(\FFbar_p\)-isomorphism
    we may suppose that \(\EC\) is defined over \(\FF_{p^2}\)
    and \(\#\EC(\FF_{p^2}) = (p + \epsilon)^2\);
    and then \(\iota\) is defined over \(\FF_{p^2}\)
    because for a supersingular elliptic curve
    over \(\FF_{p^2}\) all of the endomorphisms  are defined over \(\FF_{p^2}\).
\end{proof}

Let \(K = \QQ(\sqrt{-dp})\).
The order \(\ZZ[\sqrt{-dp}]\) has index \(2\) in \(\OK\)
if \(-dp \equiv 1\pmod{4}\),
and is equal to \(\OK\) otherwise.
If \(-dp \not\equiv 1 \pmod{4}\), then,
every natural orientation is a primitive \(\OK\)-orientation;
if \(-dp \equiv 1 \pmod{4}\),
each natural orientation is
either a primitive \(\ZZ[\sqrt{-dp}]\)-orientation or a primitive \(\OK\)-orientation.

\begin{proposition}
    \label{prop:primitivity}
    Let \((\EC,\psi)\) be a supersingular \((d,\epsilon)\)-structure
    with natural orientation \(\iota_\psi\).
    \begin{enumerate}
        \item
            If \(-dp \not\equiv 1 \pmod{4}\),
            then
            \(\iota_\psi\) is a primitive \(\OK\)-orientation.
        \item
            If \(-dp \equiv 1 \pmod{4}\),
            then
            \(\iota_\psi\) is a primitive \(\OK\)-orientation
            if the associated endomorphism \(\mu\) 
            fixes \(\EC[2]\) \emph{pointwise},
            and a primitive \(\ZZ[\sqrt{-dp}]\)-orientation otherwise.
    \end{enumerate}
\end{proposition}
\begin{proof}
    By definition, \(\iota_\psi\) is a \(\ZZ[\sqrt{-dp}]\)-orientation.
    To complete Case (2),
    it suffices to check whether 
    the element
    \(\iota_\psi(\frac{1}{2}(-1 + \sqrt{-dp})) = \frac{1}{2}(\mu - [1])\)
    of \(\End^0(\EC)\cap\iota_\psi(K)\)
    is in \(\End(\EC)\)
    (because \(\frac{1}{2}(-1 + \sqrt{-dp})\) generates \(\OK\),
    but is not in \(\ZZ[\sqrt{-dp}]\)).
    This is the case if and only if \(\mu - [1]\)
    factors over \([2]\),
    if and only if \(\mu\) fixes \(\EC[2]\) pointwise.
\end{proof}

In the light of Propositions~\ref{prop:graph-isomorphism} and~\ref{prop:primitivity},
we partition \(\Dd\)
into two subsets:
\[
    \Dd = \Ddmax \sqcup \Ddsub \,,
\]
where \(\Ddmax\)
contains the classes whose natural orientations are primitive \(\OK\)-orientations,
and \(\Ddsub\)
contains the classes 
whose natural orientations are primitive orientations by the order of
conductor 2 in~\(\OK\).
If \(-dp \not\equiv 1 \pmod{4}\),
then \(\Ddmax = \Dd\) 
and \(\Ddsub = \emptyset\).
If \(-dp \equiv 1 \pmod{4}\),
then \([\OK:\ZZ[\sqrt{-dp}]] = 2\),
so
\(\Ddmax\) resp.~\(\Ddsub\) 
consists of the \((d,\epsilon)\)-structures
where \(\mu\) acts trivially resp. nontrivially on the \(2\)-torsion.

Given Lemma~\ref{lemma:natural-induced},
\(\ell\)-isogenies of \((d,\epsilon)\)-structures 
are ``ascending'', ``descending'', and ``horizontal''
with respect to the natural orientations:
we have horizontal \(\ell\)-isogenies
between vertices in \(\Ddmax\)
and between vertices in \(\Ddsub\),
while \(\Ddmax\) and \(\Ddsub\)
are connected by ascending and descending \(2\)-isogenies.
In the language of isogeny volcanoes,
vertices in \(\Ddmax\) form the ``craters'',
and vertices in \(\Ddsub\) the ``floors''.

\subsection{The class group action}

Proposition~\ref{prop:graph-isomorphism}
translates the action of \(\Cl(\ZZ[\sqrt{-dp}])\) on \(\SS{\ZZ[\sqrt{-dp}]}\) defined above
into an action on \(\Dd\).
Theorem~\ref{theorem:action} makes this precise:
it shows that \(\Ddmax\) is a principal homogeneous space (or torsor)
under \(\Cl(\OK)\),
and that if \(\Ddsub\) is not empty
then it is a principal homogeneous space under \(\Cl(\ZZ[\sqrt{-dp}])\).

\begin{theorem}
	\label{theorem:action}
    Let \(K = \QQ(\sqrt{-dp})\),
    with maximal order \OK,
    and let \(\epsilon = \pm1\).
    \begin{itemize}
        \item
            The class group \(\Cl(\OK)\)
            acts freely and transitively on \(\Ddmax\).
        \item
            If \(\Ddsub \not= 0\),
            then \(\Cl(\ZZ[\sqrt{-dp}])\) 
            acts freely and transitively on \(\Ddsub\).
    \end{itemize}
\end{theorem}
\begin{proof}
    Let \(\OO = \OK\) or \(\ZZ[\sqrt{-dp}]\).
    Since \(p\) does not split in \(K\),
    Theorem \ref{theorem:Onuki} tells us that 
    \(\Cl(\OO)\) acts freely and transitively
    on \(\rho(\Ell(\OO)) \subseteq \SSpr{\OO}\).
    Given the isomorphism of Proposition~\ref{prop:graph-isomorphism},
    it only remains to prove that \(\rho(\Ell(\OO)) = \SSpr{\OO}\).
    For any \((\EC,\iota)\) in \(\SSpr{\OO}\),  
    Proposition~3.3 of~\cite{2021/Onuki}
    tells us that
    \((\EC,\iota)\) or \(\pconj{(\EC,\iota)}\) is in \(\rho(\Ell(\OO))\). 
    In our case,
    \emph{both} are in \(\rho(\Ell(\OO))\),
    so the action on \(\SSpr{\OO}\) is free:
    since \(\EC[\frakd] = \EC[d]\cap\ker\mu = \ker\psi\),
    the action of
    \(\frakd = (d, \sqrt{-dp})\)
    on \(\SSpr{\OO}\)
    maps \((\EC,\iota)\) to \(\pconj{(\EC,\iota)}\),
    because it maps \((\EC,\psi)\) to \(\pconj{(\EC,\psi)}\).
\end{proof}

\begin{corollary}
    Let \(K = \QQ(\sqrt{-dp})\),
    with maximal order \OK.
    If \(h_K = \#\Cl(\OK)\),
    then
    \[
        \#\Ddmax = h_K
        \qquad
        \text{and}
        \qquad
        \#\Ddsub 
        = 
        \begin{cases}
            h_K & \text{if } {-dp} \equiv 1 \pmod{8}
            \,,
            \\
            3h_K & \text{if } {-dp} \equiv 5 \pmod{8}
            \,,
            \\
            0 & \text{otherwise}
            \,.
        \end{cases}
    \]
\end{corollary}
\begin{proof}
    By Theorem~\ref{theorem:action},
    we have \(\#\Ddmax = \#\Cl(\OK)\) 
    and either \(\#\Ddsub = 0\)
    (if \(-dp \not\equiv 1 \pmod{4}\))
    or \(\#\Ddsub = \#\Cl(\ZZ[\sqrt{-dp}])\)
    (if \(-dp \equiv 1 \pmod{4}\)).
    It remains
    to compute 
    \(\#\Cl(\ZZ[\sqrt{-dp}])\)
    in the case \(-dp \equiv 1 \pmod{4}\),
    where \(\ZZ[\sqrt{-dp}]\)
    has conductor 2.
    In this case, the formula of~\cite[Theorem 7.24]{2013/Cox}
    simplifies to
    \[
        \#\Cl(\ZZ[\sqrt{-dp}]) 
        = 
        \frac{\#\Cl(\OK)}{[\OK^\times:\ZZ[\sqrt{-dp}]^\times]}
        \left(
            2 - \left(\frac{-dp}{2}\right)
        \right)
        \,,
    \]
    where the Kronecker symbol \(({-dp}/{2})\)
    is \(0\) if \(2 \mid -dp\),
    \(1\) if \(-dp \equiv \pm1 \pmod{8}\),
    and \(-1\) if \(-dp \equiv \pm3 \pmod{8}\).
    The result follows on noting that \([\OK^\times:\ZZ[\sqrt{-dp}]^\times] = 1\),
    because \(-dp\) is never \(-3\) or \(-4\).
\end{proof}

\begin{remark}
	\label{rem::Class group order}
    The Brauer--Siegel theorem states that asymptotically, \(\log_2(h_K) \sim \frac{1}{2}\log_2|\Delta_K|\),
    where 
    \(\Delta_K = -dp\) if \(-dp \equiv 1 \pmod{4}\),
    and \(-4dp\) otherwise.
    (See e.g.~\cite[Ch.~XVI]{1994/Lang} for details.)
\end{remark}

\subsection{Computing the class group action}%%%%%%%%%%%%%%%%%%%%%%%
\label{sub:computation_action}

Suppose we want to compute the action of 
(the class of) an ideal \(\frakell = (\ell,a + b\sqrt{-dp})\)
on some \((\EC,\psi)\) in \(\Dd\). 
Following~\cite{2018/De-Feo--Kieffer--Smith},
we consider two approaches: ``Vélu'' and ``modular''.

In the \textbf{``Vélu''} approach,
we compute a generator \(K_\ell\) of the kernel \(\EC[\frakell]\)
of \(\phi\):
that is, \(K_\ell\) is a point in \(\EC[\ell]\) 
such that $[a]\mu(K_\ell) = -[b]K_\ell$.
This point may only be defined over an extension $\FF_{p^{2r}}$ of $\FF_{p^2}$.
We then compute the quotient isogeny $\phi: \EC \to \EC' := \EC/\subgrp{K_\ell}$ 
using Vélu's formul\ae{},
at a cost of \(O(\ell)\) \(\FF_{p^{2r}}\)-operations,
or the algorithm of~\cite{2020/Bernstein--De-Feo--Leroux--Smith},
in \(\widetilde{O}(\sqrt{\ell})\) \(\FF_{p^{2r}}\)-operations.
Finally, we push $\psi$ through $\phi$
by computing the image of its kernel subgroup
and choosing the correct ``sign''.
If we are given an \(\FF_{p^2}\)-rational generator $G$ for \(\ker\psi\),
then pushing \(\psi\) through \(\phi\)
essentially costs one isogeny evaluation;
otherwise, this amounts to an exercise in symmetric functions,
with a cost on the order of \(O(d)\) isogeny evaluations.
Each evaluation costs~\(O(\ell)\) or~\(\widetilde{O}(\sqrt{\ell})\)
\(\FF_{p^2}\)-operations.
The total cost is dominated by the cost of the multiplication by the cofactor
$\#E(\FF_{p^{2r}}) / \ell$ needed to find $K_\ell$:
we have $\log{(\#E(\FF_{p^{2r}}) / \ell)} = 2r\log{p}$,
so constructing \(K_\ell\) requires $O(r^2 \log{p})$ operations in $\FF_{p^2}$.

To compute the action of $\frakell$ on \((\EC,\psi)\),
we compute \(G = \gcd(\Phi_d(X, X^p), \Phi_\ell(j(\EC), X))\)
(if \(d = 1\), then we take \(\Phi_1(X,X^p) = X^p - X\)).
In general \(G\) has only two roots in \(\FF_{p^2}\),
corresponding to the two \(\ell\)-neighbours.
In a non-backtracking walk we can divide by \(X - j(\EC')\),
where \((\EC',\psi')\) is the preceding vertex,
to find the next step.
Otherwise, we can distinguish between the two neighbours
by examining the action of \(\mu\) on the \(\ell\)-torsion.
Care must be taken to identify, and to appropriately handle,
the exceptional case where a neighbouring
\(j\)-invariant admits multiple \((d,\epsilon)\)-structures
modulo negation (as with the vertices \(A\) and \(C\) in the example of
Figure~\ref{fig:graph-97} below).

To compute \(\gcd(\Phi_d(X, X^p),\Phi_\ell(j(E), X))\),
compute \(F(X) := \Phi_\ell(j(\EC), X)\)
in \(O(\ell)\) \(\FF_{p^2}\)-operations,
and then
$Y := X^p \bmod F(X)$ 
using the square-and-multiply algorithm
in $O(\ell \log{p})$ \(\FF_{p^2}\)-operations.
We then compute
$Z := \Phi_d(X, Y) \mod F$,
and then \(\gcd(Z,F)\),
in \(O(d^2\ell^2)\) \(\FF_{p^2}\)-operations.
Generally \(\ell\) is polynomial in \(\log p\),
but typically it is even smaller,
and then the dominating step is the computation of \(Y\).

As in the ordinary case~\cite{2018/De-Feo--Kieffer--Smith},
the Vélu approach is more efficient when \(r^2 < \ell\);
in particular, when \(K_\ell\) is defined over \(\FF_{p^2}\).
If we are free to choose \(p\), then
we can optimize systems that use the action of a series of small primes \(\ell_i\)
by taking \(p\) such that the \(\ell_i\) split in \(\ZZ[\sqrt{-dp}]\)
\emph{and} $\ell_i \mid p+\epsilon$,
that is,
$p = c \cdot \prod_{i = 1}^n \ell_i - \epsilon$ with
$c$ a cofactor making $p$ prime.
In the case \(d = 1\),
this is exactly the optimization that is key to making CSIDH practical.

\begin{remark}
    It would be interesting to look for
    an expression for the group action
    operating directly on the parameters in 
    the Hasegawa families
    of~\S\ref{sec:Hasegawa-2}
    and~\S\ref{sec:Hasegawa-3}.
\end{remark}

\section{%%%%%%%%%%%%%%%%%%%%%%%%%%%%%%%%%%%%%%%%%%%%%%%%%%%%%%%%%%%%%%%%%%%%%%%
    The supersingular isogeny graph
}%%%%%%%%%%%%%%%%%%%%%%%%%%%%%%%%%%%%%%%%%%%%%%%%%%%%%%%%%%%%%%%%%%%%%%%%%%%%%%%
\label{sec:ssgraph}

We can now describe the structure of the isogeny graph \(\Graph{\Dd}\).
Factoring isogenies,
it suffices to describe \(\ellGraph{\Dd}\)
for prime \(\ell\).
The class group actions
of Theorem~\ref{theorem:action}
imply the isogeny counts
in Table~\ref{tab:isogeny-counts}.

\begin{table}[ht]
    \caption{The number of horizontal, ascending, and descending
        \(\ell\)-isogenies from each vertex in the \(\ell\)-isogeny graph.}
    \label{tab:isogeny-counts}
    \centering
    \begin{tabular}{r|c|c|c|c|c}
        Prime \(\ell\)
            & Conditions on \((d,p)\)
            & Vertex (sub)set 
            & Horizontal & Ascending & Descending 
        \\
        \hline
        \hline
        \multirow{6}{*}{\(\ell = 2\)} 
            & \multirow{2}{*}{\(-dp \equiv 1 \pmod{8}\)}
            & \(\Ddmax\)
            & \(2\) & \(0\) & \(1\)
        \\
            &
            & \(\Ddsub\)
            & \(0\) & \(1\) & \(0\)
        \\
        \cline{3-6}
            & \multirow{2}{*}{\(-dp \equiv 3 \pmod{8}\)}
            & \(\Ddmax\)
            & \(0\) & \(0\) & \(3\)
        \\
            & 
            & \(\Ddsub\)
            & \(0\) & \(1\) & \(0\)
        \\
        \cline{3-6}
            & \(-dp \not\equiv 1,3 \pmod{8} \)
            & \(\Dd\)
            & \(1\) & \(0\) & \(0\)
        \\
        \hline
        \(\ell > 2\) 
            & --- 
            & \(\Dd\) 
            & \(1 + (-dp/\ell)\) & 0 & 0
        \\
        \hline
    \end{tabular}
\end{table}

\paragraph{Examples.}
Figures~\ref{fig:graph-101},
\ref{fig:graph-97},
and~\ref{fig:graph-83},
display \(\ell\)-isogeny graphs
on 
\(\Dd[3][1][101]\),
\(\Dd[3][-1][97]\),
and \(\Dd[3][1][83]\)
for various \(\ell\) generating the class groups.
These figures also form examples of
the various \(2\)-isogeny structures 
listed in Table~\ref{tab:isogeny-counts}.

\begin{figure}[h]
	\centering
    \begin{tikzpicture}[thick]
        \node (A) at (0,1) {$A$} ;
        \node (mA) at (0,0) {$\neg{A}$} ;
        %\node (Ap) at (XX,YY) {$\pconj{A}$} ;
        %\node (mAp) at (XX,YY) {$\neg{\pconj{A}}$} ;

        \node (B) at (2,0) {$B$} ;
        \node (mB) at (2,1) {$\neg{B}$} ;
        \node (Bp) at (-2,1) {$\pconj{B}$} ;
        \node (mBp) at (-2,0) {$\neg{\pconj{B}}$} ;
		
        \node (C) at (-4,2) {$C$} ;
        \node (mC) at (-4,-1) {$\neg{C}$} ;
        \node (Cp) at (4,-1) {$\pconj{C}$} ;
        \node (mCp) at (4,2) {$\neg{\pconj{C}}$} ;

        \node (D) at (4,1) {$D$} ;
        \node (mD) at (4,0) {$\neg{D}$} ;
        \node (Dp) at (-4,0) {$\pconj{D}$} ;
        \node (mDp) at (-4,1) {$\neg{\pconj{D}}$} ;

        \node (E) at (2,2) {$E$} ;
        \node (mE) at (2,-1) {$\neg{E}$} ;
        \node (Ep) at (-2,-1) {$\pconj{E}$} ;
        \node (mEp) at (-2,2) {$\neg{\pconj{E}}$} ;

        \node (F) at (0,2) {$F$} ;
        \node (mF) at (0,-1) {$\neg{F}$} ;
        %\node (Fp) at (XX,YY) {$\pconj{F}$} ;
        %\node (mFp) at (XX,YY) {$\neg{\pconj{F}}$} ;

        % the cycle:
        \draw (A) -- (mB) ;
        \draw (mB) -- (D) ;
        \draw (D.east) edge[bend left=90] (mD.east) ;
        \draw (mD) -- (B) ;
        \draw (B) -- (mA) ;
        \draw (mA) -- (mBp) ;
        \draw (mBp) -- (Dp) ;
        \draw (Dp.west) edge[bend left=90] (mDp.west) ;
        \draw (mDp) -- (Bp) ;
        \draw (Bp) -- (A) ;

        % hairs:
        \draw (A) -- (F) ;
        \draw (mB) -- (E) ;
        \draw (D) -- (mCp) ;
        \draw (mD) -- (Cp) ;
        \draw (B) -- (mE) ;
        \draw (mA) -- (mF) ;
        \draw (mBp) -- (Ep) ;
        \draw (Dp) -- (mC) ;
        \draw (mDp) -- (C) ;
        \draw (Bp) -- (mEp) ;
	\end{tikzpicture}
	\caption{\ellGraph[2]{\Dd[3][1][101]} for $\ell = 2$. 
        The class group of \(\QQ(\sqrt{-303})\) is isomorphic to \(\ZZ/10\ZZ\),
        and generated by an ideal over \(2\)
        (we see this in the length-10 cycle).
        The correspondence between vertex labels and parameters
        for the degree-3 Hasegawa family of~\S\ref{sec:Hasegawa-3}
        (with \(\Delta = 2\))
        is
        \(A \leftrightarrow 0\),
        \(B \leftrightarrow 6\),
        \(C \leftrightarrow 24\),
        \(D \leftrightarrow 25\),
        and
        \(E \leftrightarrow 42\);
        the special vertex \(F\), which has no Hasegawa parameter,
        is \((\EC,\psi)\) with \(\EC: y^2 = x^3 + 1\)
        and \(\psi: (x,y) \mapsto ((67x^3 + 66)/x^2, (89x^3 + 96)\sqrt{2}y/x^3)\).
        Note that \(\pconj{A} = \neg{A}\) and \(\pconj{F} = \neg{F}\).
        The underlying curves of \(B\) and \(C\) are isomorphic.
    }
    \label{fig:graph-101}
\end{figure}
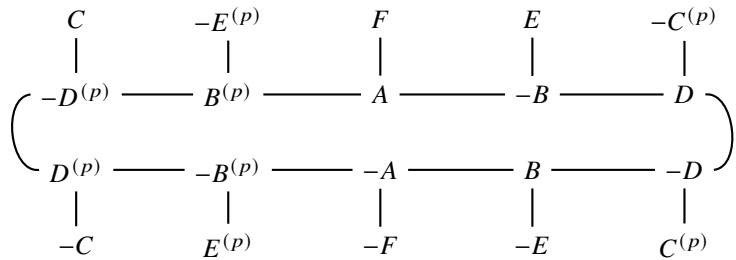

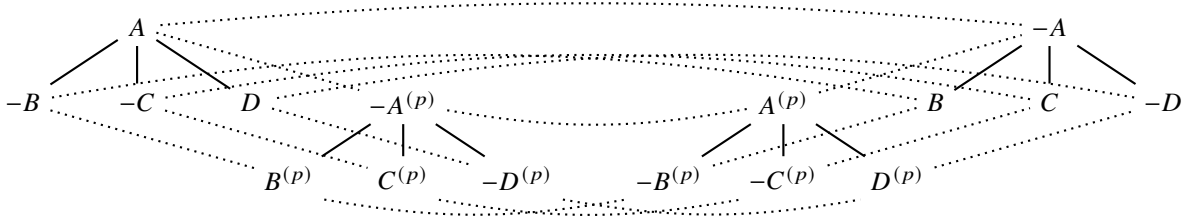
\begin{figure}
    \centering
    \begin{tikzpicture}[node distance={40mm}, thick] 
            \node (A)  at (-6,0) {$A$};  
            \node (mB)  at (-7.5, -1) {$\neg{B}$};
            \node (mC)  at (-6, -1) {$\neg{C}$}; 
            \node (D)  at (-4.5, -1) {$D$};

            \draw (A) -- (mB) ;
            \draw (A) -- (mC) ;
            \draw (A) -- (D) ;
            
            \node (mAp)  at (-2.5,-1) {$\neg{\pconj{A}}$};  
            \node (Bp)   at (-4,-2) {$\pconj{B}$};
            \node (Cp)   at (-2.5,-2) {$\pconj{C}$}; 
            \node (mDp)  at (-1,-2) {$\neg{\pconj{D}}$};
            
            \draw (mAp) -- (Bp) ;
            \draw (mAp) -- (Cp) ;
            \draw (mAp) -- (mDp) ;

            \node (Ap)  at (2.5,-1) {$\pconj{A}$};
            \node (mBp) at (1,-2) {$\neg{\pconj{B}}$};
            \node (mCp) at (2.5,-2) {$\neg{\pconj{C}}$}; 
            \node (Dp)  at (4,-2) {$\pconj{D}$};
            
            \draw (Ap) -- (mBp) ;
            \draw (Ap) -- (mCp) ;
            \draw (Ap) -- (Dp) ;

            \node (mA)  at (6,0) {$\neg{A}$};  
            \node (B)   at (4.5,-1) {$B$};
            \node (C)   at (6,-1) {$C$}; 
            \node (mD)  at (7.5,-1) {$\neg{D}$}; 

            \draw (mA) -- (B) ;
            \draw (mA) -- (C) ;
            \draw (mA) -- (mD) ;

            \draw[dotted] (A) -- (mAp) ;
            \draw[dotted] (mAp) edge[bend right=10] (Ap) ;
            \draw[dotted] (Ap) -- (mA) ;
            \draw[dotted] (mA) edge[bend right=5] (A) ;

            \draw[dotted] (B) edge[bend right=10] (mB) ;
            \draw[dotted] (mB) -- (Bp) ;
            \draw[dotted] (Bp.south east) edge[bend right=10] (mBp.south west) ;
            \draw[dotted] (mBp) -- (B) ;

            \draw[dotted] (C) edge[bend right=10] (mC) ;
            \draw[dotted] (mC) -- (Cp) ;
            \draw[dotted] (Cp.south east) edge[bend right=10] (mCp.south west) ;
            \draw[dotted] (mCp) -- (C) ;

            \draw[dotted] (mD) edge[bend right=10] (D) ;
            \draw[dotted] (D) -- (mDp) ;
            \draw[dotted] (mDp.south east) edge[bend right=10] (Dp.south west) ;
            \draw[dotted] (Dp) -- (mD) ;
    \end{tikzpicture}
    \caption{The isogeny graphs
        \(\ellGraph[2]{\Dd[3][-1][97]}\) (solid) 
        and
        \(\ellGraph[5]{\Dd[3][-1][97]}\) (dotted).
        We have \(\Cl(\QQ(\sqrt{-3\cdot97})) \cong \ZZ/4\ZZ\),
        generated by an ideal over~\(5\).
        The \(2\)-isogenies are ascending/descending up/down the page;
        the \(5\)-isogenies are horizontal.
        The correspondence between vertex labels and parameters for the degree-3
        Hasegawa family of~\S\ref{sec:Hasegawa-3} (with \(\Delta = 5\))
        is 
        \(A \leftrightarrow 47\), 
        \(B \leftrightarrow 1\), 
        \(C \leftrightarrow 14\), 
        and 
        \(D \leftrightarrow 22\).
        The underlying curves of \(A\) and \(C\) are isomorphic.
    }
    \label{fig:graph-97}
\end{figure}

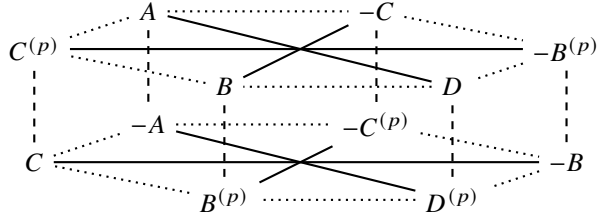
\begin{figure}[ht]
    \centering
    \begin{tikzpicture}[node distance={40mm}, thick] 
        \node (1)  at (-3.5,0) {$\pconj{C}$};
        \node (2)  at (-2, 0.5) {$A$};
        \node (3)  at (1, 0.5) {$\neg{C}$};
        \node (4)  at (3.5, 0) {$\neg{\pconj{B}}$};
        \node (5)  at (2, -0.5) {$D$};
        \node (6)  at (-1,-0.5) {$B$};

        \node (11)  at (-3.5,-1.5) {$C$};
        \node (12)  at (-2, -1) {$\neg{A}$};
        \node (13)  at (1, -1) {$\neg{\pconj{C}}$};
        \node (14)  at (3.5, -1.5) {$\neg{B}$};
        \node (15)  at (2, -2) {$\pconj{D}$};
        \node (16)  at (-1,-2) {$\pconj{B}$};

        \draw[dotted] (1) -- (2);
        \draw[dotted] (2) -- (3);
        \draw[dotted] (3) -- (4);
        \draw[dotted] (4) -- (5);
        \draw[dotted] (5) -- (6);
        \draw[dotted] (6) -- (1);
            
        \draw[dotted] (11) -- (12);
        \draw[dotted] (12) -- (13);
        \draw[dotted] (13) -- (14);
        \draw[dotted] (14) -- (15);
        \draw[dotted] (15) -- (16);
        \draw[dotted] (16) -- (11);        
            
        \draw (1) -- (4);
        \draw (2) -- (5);
        \draw (3) -- (6);
        \draw (11) -- (14);
        \draw (12) -- (15);
        \draw (13) -- (16);
            
        \draw[dashed] (1) -- (11);
        \draw[dashed] (2) -- (12);
        \draw[dashed] (3) -- (13);
        \draw[dashed] (4) -- (14);
        \draw[dashed] (5) -- (15);
        \draw[dashed] (6) -- (16);
    \end{tikzpicture}
    \caption{\ellGraph{\Dd[3][1][83]} for $\ell = 2$ (solid), $\ell = 3$ (dashed) and $\ell = 5$ (dotted).
    All isogenies are horizontal.
    We have \(\Cl(\QQ(\sqrt{-3\cdot83})) \cong \ZZ/2\ZZ\times\ZZ/6\ZZ\),
    with the \(\ZZ/2\ZZ\)-factor generated by the ideal above 3,
    and the \(\ZZ/6\ZZ\)-factor generated by an ideal above 5
    (we see this in the length-6 cycles);
    the ideal above 2 is the cube of an ideal above 5.
    The correspondence between vertex labels and parameters for the
    degree-3 Hasegawa family of~\S\ref{sec:Hasegawa-3}
    (with \(\Delta = 2\)) is
    \(A \leftrightarrow 0\),
    \(B \leftrightarrow 32\),
    \(C \leftrightarrow 40\);
    the special vertex \(D\),
    which has no Hasegawa parameter,
    is \((\EC:y^2 = x^3 + 1,\psi)\)
    where \(\psi\) maps \((x,y)\)
    to \((((72\sqrt{2} + 14)x^3 + (39\sqrt{2} + 56))/x^2, \sqrt{2}(35x^3 + 52)y/x^3)\).
    Note that \(\neg{A} = \pconj{A}\).
    }
    \label{fig:graph-83}
\end{figure}

\paragraph{Involutions.}
There are two obvious involutions on \(\Graph{\Dd}\),
negation and conjugation.
These are generally not the only involutions.
Every prime \(\ell\) dividing the discriminant ramifies
in \(\OK\) (and \(\ZZ[\sqrt{-dp}]\));
the prime \(\frakell\) over \(\ell\)
gives an element of order \(2\)
in \(\Cl(\OK)\) (and \(\Cl(\ZZ[\sqrt{-dp}])\)),
and thus an involution on \(\Graph{\Dd}\).
Let
\(\frakd_1,\ldots,\frakd_n\) be the primes above
the prime factors of \(d\),
and \(\frakp\) the prime above \(p\);
note that
\([\frakd_1]\cdots[\frakd_n] = [\frakp]\),
because \(\frakd_1\cdots\frakd_n\frakp = (\mu)\).
If \(-dp \equiv 1\) or \(2 \pmod{4}\)
then
\(
    \Cl(\OK)[2] = \subgrp{[\frakd_1],\ldots,[\frakd_n],[\frakp]}
\),
so \(\Cl(\OK)[2] \cong (\ZZ/2\ZZ)^n\).
If \(-dp \equiv 3 \pmod{4}\),
then
\(
    \Cl(\OK)[2] = \subgrp{[\fraka],[\frakd_1],\ldots,[\frakd_n],[\frakp]}
\)
where \(\fraka\) is the ideal above \(2\),
and \(\Cl(\OK)[2] \cong (\ZZ/2\ZZ)^{n+1}\).
In each case,
the action of the ideal class
\(\prod_i[\frakd_i] = [\frakp]\)
on any \((d,\epsilon)\)-structure \((\EC,\psi)\)
is realised by the isogeny \(\psi: (\EC,\psi) \to
(\pconj{\EC},\pconj{\psi})\),
and is therefore equal to the conjugation involution.

Since the group actions are free,
each of the involutions that come from nontrivial 2-torsion elements in
the class groups---including conjugation---has no fixed points.
Negation, on the other hand,
can have fixed points:
for example, if \(p \equiv 3 \pmod{4}\) 
and \(\EC\) is the curve with \(j\)-invariant 1728,
and \(i\) is an automorphism of degree 4,
then \((\EC,i)\) is a \((1,1)\)-structure,
and \((\EC,i) \cong (\EC,-i)\).
This is the only fixed point among \((1,1)\)-structures,
and its existence is implied by the fact that
the class number of \(\Cl(\sqrt{-p})\)
is odd when \(p \equiv 3 \pmod{4}\).

\begin{remark}
    If \(-dp \equiv 5 \pmod{8}\),
    then there is an order-3 automorphism $T$ of \(\Ddsub\)
    cycling the triplets of vertices with ascending \(2\)-isogenies
    to the same vertex in \(\Ddmax\).
    We will see that $T$
    is induced by the action of an ideal class in~\(\Cl(\ZZ[\sqrt{-dp}])\).
    The ideal \(\frakt = (4,\sqrt{-dp}-1)\ZZ[\sqrt{-dp}]\)
    has order \(3\) in \(\Cl(\ZZ[\sqrt{-dp}])\),
    but capitulates to become the principal ideal \((2)\) in \(\OK\)
    (because \(\sqrt{-dp}-1 = 2\omega\),
    where \(\omega\) is the unit \(\frac{1}{2}(\sqrt{-dp}-1)\));
    indeed, \(\frakt\) generates the kernel
    of the canonical homomorphism \(\Cl(\ZZ[\sqrt{-dp}]) \to \Cl(\OK)\).
    Since \(\frakt\) meets the conductor,
    its action on \(\Ddsub\) is not well-defined,
    but we can consider the action of an equivalent ideal in the class
    group.
    Let \(\prod_i\ell_i^{e_i}\) be the prime factorization of \((dp +
    1)/4\) (and note that each \(\ell_i\) is odd);
    then
    \(
        (\sqrt{-dp}-1)
        =
        \frakt
        \cdot
        \prod_i
        \frakell_i^{e_i}
        \)
        where 
        \(
        \frakell_i := (\ell_i,\sqrt{-dp}-1)
    \);
    the product \(\prod_i\frakell_i^{e_i}\)
    is equivalent to \(\frakt\) in \(\Cl(\ZZ[\sqrt{-dp}])\),
    prime to the conductor,
    and its action on $\Ddsub$ induces the automorphism $T$.
    In the case where \(d = 1\) (CSIDH),
    this is explained at length in~\cite{2020/Onuki--Takagi}.
\end{remark}

\paragraph{Crossroads.}
The map \((\EC,\psi) \mapsto \EC\)
defines a covering from \(\Graph{\Dd}\) onto a subgraph of 
the isogeny graph of all supersingular curves over \(\FF_{p^2}\).
For \(d_1\not=d_2\)
the images of \(\Graph{\Dd[d_1]}\) and \(\Graph{\Dd[d_2]}\) 
can intersect,
forming ``crossroads'' where we can switch from walking in
\(\Graph{\Dd[d_1]}\) into \(\Graph{\Dd[d_2]}\),
and vice versa.

\begin{definition}
    Let \(d_1 \not= d_2\) be squarefree integers
    such that \(d_1d_2\) is squarefree.
    We say that a supersingular curve \(\EC/\FF_{p^2}\)
    with \(\#\EC(\FF_{p^2}) = (p + \epsilon)^2\)
    is a \emph{\((d_1,d_2)\)-crossroad}
    if there exist isogenies \(\psi_1: \EC \to \EC_1\)
    and \(\psi_2: \EC \to \EC_2\)
    such that \((\EC,\psi_1)\) is a \((d_1,\epsilon)\)-structure
    and \((\EC,\psi_2)\) is a \((d_2,\epsilon)\)-structure.
\end{definition}

If \((\EC,\psi)\) is a \((d_1,\epsilon)\)-structure,
then we can easily check whether \(\EC\) is a \((d_1,d_2)\)-crossroad
by evaluating the classical modular polynomial \(\Phi_{d_2}\) at
\((j(\EC,),j(\EC)^p)\).
However, \((d_1,d_2)\)-crossroads are generally very rare.
Indeed, if \(\EC\) is a \((d_1,d_2)\)-crossroad,
then it has an endomorphism of degree \(d_1d_2\) with cyclic kernel
(in particular,
\((d_1,d_2)\)-crossroads with \(d_1d_2 < \frac{\sqrt{p}}{2}\)
appear in the isogeny ``valleys'' described in \cite{2019/Love--Boneh}).
We can therefore enumerate the entire set of \((d_1,d_2)\)-crossroads
over a given \(\FF_{p^2}\)
by computing the set of roots \(j\) of \(\Phi_{d_1d_2}(x,x)\)
in \(\FF_{p^2}\),
and then checking for which \(j\) we have \(\Phi_{d_1}(j,j^p) = 0\).
The polynomial \(\Phi_{d_1d_2}(x,x)\) has
degree \(\prod_\ell (\ell+1)\)
where \(\ell\) ranges over the prime factors of \(d_1d_2\),
so there are only \(O(d_1d_2)\) \((d_1,d_2)\)-crossroads
(up to isomorphism)
among the \(\OO (\sqrt{dp})\) vertices in \(\Graph{\Dd[d_1]}\).

But while crossroads are rare,
computing the few examples is relatively easy,
and computing \((d_1,d_2)\)-crossroads
gives us a useful way of quickly constructing some vertices
in \(\Graph{\Dd[d_1]}\) (and in \(\Graph{\Dd[d_2]}\)).
Suppose we want to construct a vertex in \(\Graph{\Dd[d_1]}\).
Since the vertices in \(\Graph{\Dd[d_1]}\) correspond to curves with an endomorphism subring
isomorphic to \(\ZZ[\sqrt{-d_1p}]\), we might try to construct a vertex from 
a root in \(\FF_{p^2}\) of the Hilbert class polynomial for \(\QQ(\sqrt{-d_1p})\);
but the degree of this polynomial, which is the order of the class
group, is exponential with respect to \(\log p\),
so this approach is infeasible for large \(p\).
Instead, we choose a small squarefree \(d_2\)
such that \(p\) does not split 
in the maximal order of \(\QQ(\sqrt{-d_1d_2})\).
If there exists a \((d_1,d_2)\)-crossroad \(\EC/\FF_{p^2}\),
then its \(j\)-invariant is
a root in \(\FF_{p^2}\) 
of a quadratic factor of the Hilbert class polynomial for
\(\QQ(\sqrt{-d_1d_2})\),
because the composition of the \(d_1\)-isogeny \(\EC \to \pconj{\EC}\)
with the conjugate \(d_2\)-isogeny \(\pconj{\EC} \to \EC\)
is a cyclic endomorphism of degree \(d_1d_2\).
All other vertices in \(\Graph{\Dd[d_1]}\) can then be reached through the class group action.

\section{%%%%%%%%%%%%%%%%%%%%%%%%%%%%%%%%%%%%%%%%%%%%%%%%%%%%%%%%%%%%%%%%%%%%%%%
    Cryptographic applications
}%%%%%%%%%%%%%%%%%%%%%%%%%%%%%%%%%%%%%%%%%%%%%%%%%%%%%%%%%%%%%%%%%%%%%%%%%%%%%%%
\label{sec:crypto}

The action of 
\(\Cl(\OK)\) on $\Ddmax$ and \(\Cl(\ZZ[\sqrt{-dp}])\) on $\Ddsub$ 
makes \(\Graph{\Dd}\) a natural candidate setting for 
group-action/HHS-based
postquantum cryptosystems
following Stolbunov~\cite{2006/Rostovtsev--Stolbunov,2009/Stolbunov,2010/Stolbunov} 
and
Couveignes~\cite{2006/Couveignes}.
For example,
for each \(d > 1\),
we can define a key exchange algorithm on \(\Dd\)
generalizing
CSIDH~\cite{2018/Castryck--Lange--Martindale--Panny--Renes},
which uses the action of \(\Cl(\ZZ[\sqrt{-p}])\) on \(\Ddsub[1][1]\)
and CSURF~\cite{2020/Castryck--Decru},
which uses the action of \(\Cl(\QQ(\sqrt{-p}))\) on \(\Ddmax[1][1]\).
Despite the prominence of orientations, 
the relationship between key exchange in \(\Dd\)
and the OSIDH protocol~\cite{2020/Colo--Kohel}
is distant.
The \(\OO\)-orientations in OSIDH
involve orders \(\OO\)
with massive conductors in \(\OK\)
where \(\OK\) has tiny class number;
here, \(\OO\)
has tiny conductor and \(\OK\)
has massive class number.

\subsection{Hard problems}
\label{sec:hard-problems}

The conjectural hard problems for the action of \(\Cl(\OK)\) on \(\Dd\)
are vectorization (the analogue of the DLP) and parallelization (the
analogue of the CDHP) from Couveigne's \emph{Hard Homogenous Spaces}
framework~\cite{2006/Couveignes}.

\begin{definition}[Vectorization]
	\label{vectorization}
	Given $(\EC, \psi)$ and $(\EC', \psi')$ in $\Dd$,
    find $\fraka \in \Cl(\OK)$ such that $\fraka \cdot (\EC, \psi) = (\EC', \psi')$.
\end{definition}

\begin{definition}[Parallelization]
	\label{parallelization}
	Given $(\EC_0, \psi_0)$, $(\EC_1, \psi_1)$, and $(\EC_2, \psi_2)$
    in $\Dd$,
    compute the unique $(\EC_3,\psi_3)$ in $\Dd$
    such that $(\EC_3,\psi_3) = (\fraka_1\fraka_2)\cdot(\EC_0,\psi_0)$
    where $(\EC_i,\psi_i) = \fraka_i\cdot(\EC_0,\psi_0)$
    for $i = 1$ and $2$.
\end{definition}

Solving Vectorization immediately solves Parallelization.
In the opposite direction,
no classical reduction is known,
but the quantum equivalence of these two problems is shown in~\cite{2021/Galbraith--Panny--Smith--Vercauteren}.

An extensive study of the possible classical and quantum attacks on
Vectorization for \(d = 1\)
can be found in~\cite{2018/Castryck--Lange--Martindale--Panny--Renes};
all of these attacks extend to \(d > 1\)
with a slowdown at most polynomial in \(d\)
for class groups of the same size,
with that slowdown due to potentially more complicated
isogeny evaluation and comparison algorithms.
The best classical attack known on Vectorization
is to use random walks in \(\Graph{\Dd}\), exactly as in the \(d=1\) case in~\cite{2016/Delfs--Galbraith},
which gives a solution after an expected \(O((dp)^{1/4})\) isogeny steps.
Since Vectorization is an instance of
the Abelian Hidden Shift Problem, 
the best quantum attack is Kuperberg's algorithm~\cite{2005/Kuperberg,2004/Regev,2013/Kuperberg}
using the Childs--Jao--Soukharev quantum isogeny-evaluation algorithm as a
subroutine~\cite{2014/Childs--Jao--Soukharev},
adapted to push \(\psi\) through the \(\ell\)-isogenies.
This adaptation may incur a practically significant but asymptotically
negligible cost;
the result is a subexponential algorithm running in time $L_{dp}[1/2, \sqrt{2}]$.
Even for \(d = 1\),
there is some debate as to the concrete cost of this quantum
algorithm, and the size of \(p\)
required to provide a cryptographically hard problem instance
for common security levels~\cite{2019/Bernstein--Lange--Martindale--Panny,
                                 2020/Bonnetain--Schrottenloher,
                                 2020/Peikert}.
(If and) when some consensus forms on secure parameter sizes for CSIDH,
the same parameter sizes should make
Vectorization and Parallelization in \(\Dd\) cryptographically hard, too.

We should also consider the impact of the various involutions on \(\Graph{\Dd}\).
The negation involution already exists for \(d = 1\),
where it essentially flips between a curve and its quadratic twist over \(\FF_p\).
This involution has not yet been exploited to give an interesting speedup in
solving Vectorization or Parellization in the case \(d = 1\);
a speedup for any \(d\) would be an interesting result.
For $d > 1$, however, there is at least one new involution: namely, conjugation.
We note that solving Vectorization modulo conjugation
solves Vectorization,
because a vertex and its conjugate are always connected
by the action of an ideal of norm \(d\).
Working modulo conjugation
allows us to shrink search spaces by a factor of \(2\),
yielding a speedup by a factor of up to \(\sqrt{2}\)
analogous to working modulo negation when solving the classical ECDLP
(as in~\cite{2011/Bernstein--Lange--Schwabe}).
When \(d\) has \(n\) prime factors,
we get more involutions
that would allow us to work with equivalence classes of \(2^n\)
vertices,
shrinking the search spaces by a factor of \(2^n\).
Prime \(d\) therefore seems the simplest and strongest case to us.

Finally, we note that if a random walk should wander into a crossroad,
then we have found an isogeny to a supersingular curve with 
much known on its endomorphism ring.
In this case, attacks analogous to that
of~\cite{2016/Galbraith--Petit--Shani--Ti} should apply.
But as we have seen, crossroads are vanishingly rare;
their existence should not create any weakness for schemes based on \(\Graph{\Dd}\),
no more than they do for CSIDH.
	
\subsection{Non-interactive key exchange}
\label{sec:NIKE}

We now describe a non-interactive key-exchange protocol based on
the class group action on \(\Graph{\Dd}\),
generalizing CSIDH (the case \(d = 1\)).
The public parameters are a prime \(p\),
a prime \(d\),
an \(\epsilon\) in \(\{1,-1\}\),
a set of primes \(\{\ell_i\}_{i=1}^n\)
prime to \(dp\) and splitting in \(\QQ(\sqrt{-dp})\),
together with a prime ideal \(\frakell_i\)
above each \(\ell_i\),
and a ``starting'' vertex \((\EC_0,\psi_0)\) in \(\Dd\)
(constructed using the crossroad technique, for example).
We also fix a secret keyspace \(\mathcal{K} \subset \ZZ^n\) of exponent
vectors
such that \(\#\mathcal{K} \ge 2^{2\lambda}\) 
to provide \(\lambda\) bits of security against meet-in-the-middle attacks
(though smaller \(\mathcal{K}\) may suffice: see~\cite{2020/CSCDJRH}).
The prime \(p\) must be large enough
that Vectorization and Parallelization
cannot be solved in fewer than \(2^\lambda\) classical operations,
or a comparable quantum effort.

For key generation,
each user randomly samples their private key as a vector $(e_i)_{1 \le i \le n}$
from $ \mathcal{K} $,
representing
the ideal class $[\fraka] = [\prod_{i = 1}^n \mathfrak{l}_i^{e_i}]$
in $\Cl(\OK) $.
Their public key is a vertex $(\EC, \psi)$ 
representing \([\fraka]\cdot(\EC_0,\psi_0)\),
which we can compute using the methods
of~\S\ref{sub:computation_action}.
The public key may be compressed to a single element of \(\FF_p\) plus a few bits
using the modular techniques of~\S\ref{sec:modular}.

For key exchange,
suppose Alice and Bob have key pairs 
$([\fraka], (\EC_A, \psi_A))$ and 
$([\frakb], (\EC_B, \psi_B))$,
respectively.
Alice receives and validates $(\EC_B,\psi_B)$,
and computes 
$S_{AB} = (\EC_{AB},\psi_{AB}) = [\fraka]\cdot (\EC_B,\psi_B)$;
Bob receives and validates $(\EC_A,\psi_A)$,
and computes
$S_{BA} = (\EC_{BA},\psi_{BA}) = [\frakb]\cdot (\EC_A,\psi_A)$.
The commutativity of the group action 
implies that \(S_{AB} \cong S_{BA}\),
so Alice and Bob have a shared secret \emph{up to isomorphism}.
To obtain a unique shared value for cryptographic key derivation,
they can derive 
a modular ``compressed'' representation of the shared secret
as in~\S\ref{sec:modular}
(for example, when \(d = 2\) or \(3\),
the parameter \(u\) for the family of~\S\ref{sec:Hasegawa-2}
or~\S\ref{sec:Hasegawa-3} and a sign bit suffice),
or simply take \(j(\EC_{AB}) = j(\EC_{BA})\)
with a minimal security loss.

\begin{remark}
    When ideal classes represent cryptographic secrets,
    it is important to compute their actions in constant time.
    A number of techniques have been proposed for this
    in the context of CSIDH
    \cite{2019/Meyer--Campos--Reith,
          2020/Onuki--Aikawa--Yamazaki--Takagi,
          2019/Cervantes-Vazquez--Chenu--Chi-Dominguez--De-Feo--Rodriguez-Henriquez--Smith,
          2020/Campos--Kannwischer--Meyer--Onuki--Stottinger,
          2021/Banegas--Bernstein--Campos--Chou--Lange--Meyer--Smith--Sotakova}.
    Each of these methods generalizes in a straightforward way
    to compute class-group actions on \((d,\epsilon)\)-structures.
    The only real algorithmic difference
    when evaluating an isogeny \(\phi: (\EC,\psi) \to (\EC',\psi')\)
    is that the isogeny \(\psi\) must be pushed through \(\phi\) in constant-time as well.  
    For $d = 2$ and $3$, this amounts to pushing the $x$-coordinate of a single point through the isogeny,
    something that is already part of constant-time CSIDH implementations.
    For $d > 3$ the kernel polynomial of \(\psi\) can be pushed
    through \(\phi\) using symmetric functions.
\end{remark}

\subsection{Key validation and supersingularity testing}%%%%%%%%%%%%%%%%%%%%%%%%
\label{subsection::Validation}

Public key validation is an important step
in many public-key cryptosystems,
notably in non-interactive key exchanges
where it is a defence against active attacks.
In our situation, this amounts to proving that a pair \((\EC,\psi)\)
represents an element of \(\Dd\) (or \(\Ddmax\), or \(\Ddsub\)).
The first step is to check that \((\EC,\psi)\) is a \((d,\epsilon)\)-structure:
specifically, we must check that \(\psi\) is indeed an isogeny from \(\EC\) to \(\pconj{\EC}\)
and that \(\dualof{\psi} = \epsilon\pconj{\psi}\).
This can be done with two \(d\)-isogeny computations, % with Vélu's formul\ae{},
which costs very little when \(d\) is small.

Verifying supersingularity is more complicated.
For \(d = 1\) (CSIDH),
we just check whether a curve over \(\FF_p\) has order \(p+1\),
which can be done efficiently by probabilistically generating a point of
order \(m\mid p+1\) with \(m > 4\sqrt{p}\)
(see~\cite[\S5]{2018/Castryck--Lange--Martindale--Panny--Renes}).
But this technique does not extend to \(d > 1\),
where we must check if \(\EC/\FF_{p^2}\) has \((p+\epsilon)^2\) points:
our valid curves have \(\EC(\FF_{p^2}) \cong (\ZZ/(p+\epsilon)\ZZ)^2\), 
and therefore no points with the required order \(> 4p\).

Instead, for \(d > 1\) we can specialize the deterministic supersingularity
test of Sutherland~\cite{2012/Sutherland}.
Let \(\pi_\EC\) be the Frobenius endomorphism of \(\EC/\FF_{p^2}\).
The discriminant of \(\ZZ[\pi_\EC]\) is bounded by \(4p^2\),
so the conductor of \(\ZZ[\pi_\EC]\) in \(\OK\) is bounded by \(2p\);
hence, if \(\EC\) is ordinary,
then the maximal height of the \(2\)-isogeny volcano containing \(\EC\) is \(\log_2(p)+1\).
Sutherland's supersingularity test takes
random non-backtracking \(2\)-isogeny walks 
starting from each of the three \(2\)-isogeny neighbours of \(\EC\).
If \(\EC\) is ordinary,
then at least one of these walks will descend the \(2\)-isogeny volcano,
and will therefore terminate 
(with no non-backtracking step defined over \(\FF_{p^2}\)) 
after at most \(\log_2(p)+1\) steps.
Conversely,
if no walk terminates after \(\log_2(p)+1\) steps,
then \(\EC\) must be supersingular.

In our case,
we know that \(\End(\EC) \supset \ZZ[\mu] \supset \ZZ[\pi_\EC]\),
and the conductor of \(\ZZ[\pi_\EC]\) in \(\ZZ[\mu]\)
is the integer \(|r|\) of Proposition~\ref{prop:r},
which is bounded by \(2\sqrt{p/d}\).
We can therefore reduce the walk length limit from \(\log_2(p)+1\)
to \(\frac{1}{2}(\log_2(p)-\log_2(d)) + 1\).
We can also use the fact that \(\ZZ[\mu] \subset \End(\EC)\)
to ensure that we choose a ``descending'' path
within at most two steps, and omit the other two paths.
Thus, we can determine if a \((d,\epsilon)\)-structure \((\EC,\psi)\) is
supersingular for the cost of computing two \(d\)-isogenies
and \( (\frac{1}{2}(\log_2(p)-\log_2(d)) + 5) \) \(2\)-isogenies.

We can determine whether \((\EC,\psi)\) is in \(\Ddmax\) or \(\Ddsub\)
(if required, and only if \(-dp \equiv 1 \pmod{4}\))
by computing the action of \(\mu\) on the \(2\)-torsion
(at the cost of one or two \(d\)-isogeny evaluations)
or by computing the \(2\)-neighbours of \((\EC,\psi)\) in
\(\ellGraph[2]{\Dd}\).

\subsection{Generalized Delfs--Galbraith algorithms}
\label{sec:Delfs--Galbraith}

Let \(S_p\) be the set of supersingular curves over \(\FF_{p^2}\),
up to isomorphism.
The general supersingular isogeny problem
is, given \(\EC_1\) and \(\EC_2\) in \(S_p\),
to compute an isogeny \(\phi: \EC_1 \to \EC_2\).

In~\cite{2016/Delfs--Galbraith}, Delfs and Galbraith use the subset
of supersingular curves defined over $\FF_{p}$,
which we can identify with \(\Dd[1][1]\),
to improve classical isogeny-finding algorithms based on random walks.
Their algorithm has two phases:
\begin{enumerate}
    \item
        Compute a random non-backtracking isogeny walk from \(\EC_1\) resp. \(\EC_2\)
        until we land on a curve \(\EC_1'\) resp. \(\EC_2'\) in \(\Dd[1][1]\).
        These walks yield isogenies \(\phi_1: \EC_1 \to \EC_1'\)
        and \(\phi_2: \EC_2 \to \EC_2'\).
        The isogeny graph on \(S_p\) has excellent mixing properties,
        and since \(\#S_p \approx p/12\) and \(\#\Dd[1][1] = O(\sqrt{p})\),
        this first phase takes an expected \(O(\sqrt{p})\) random isogeny steps.
    \item
        Find an isogeny \(\phi': \EC_1' \to \EC_2'\)
        using the action of \(\Cl(\QQ(\sqrt{-p}))\) acting on \(\Dd[1][1]\)
        (that is, solve Vectorization with \(d = 1\)).
        Under the Generalized Riemann Hypothesis,
        \(\Cl(\QQ(\sqrt{-p}))\)
        is generated by the set \(\mathcal{L}\) of ideals of prime norm up to $6 \log{(|\Delta|)^2}$,
        where \(\Delta\) is the discriminant of \(\QQ(\sqrt{-p})\)
        (see~\cite{1984/Bach})
        though in practice we do not need so many primes.
        The \(\mathcal{L}\)-isogeny graph on \(\Dd[1][1]\)
        is therefore connected,
        and we can use random walks in this subgraph to construct \(\phi'\).
        By the birthday paradox, this phase takes an expected \(O(\sqrt[4]{p})\)
        random steps before finding the collision yielding \(\phi'\).
\end{enumerate}

The Delfs--Galbraith algorithm exploits
the action of \(\Cl(\QQ(\sqrt{-p}))\) on \(\Dd[1][1]\)
to solve the isogeny problem in~\(S_p\).
We can generalize their algorithm
by replacing the distinguished subgraph \(\Graph{\Dd[1]}\)
with a union of subgraphs \(\sqcup_{d\in D}\Graph{\Dd[d]}\)
where \(D\) is a set of coprime squarefree integers prime to \(p\).
In Phase 1,
we now take random walks from \(\EC_1\) and \(\EC_2\)
into \(\sqcup_{d\in D}\Dd[d]\).%
\footnote{
    To measure the feasability of this attack, we need to estimate the
    average number of steps from a general supersingular elliptic curve
    \(\EC/\FF_{p^2}\) to a curve in (the image of) $\Dd$.
    This distance follows a binomial law $(m, \mathcal{P})$ where $m$ is the number of steps and $\mathcal{P} = \sqrt{d/p}$.
    Hence, the probability $\mathbb{P}(X > 1)$ that we reach at least one
    element in $\Dd$ after $m$ steps from $\EC$ is 
    $\mathbb{P}(X > 1) = 1 - \mathbb{P}(X = 0) = 1 - (1 - \sqrt{d/p})^m$.
    When $d = 1$, this addresses some of the heuristic observations
    in~\cite{2021/Arpin--Camacho-Navarro--Lauter--Lim--Nelson--Scholl--Sotakova},
    notably the distance to the $\FF_p$-spine.
}
In Phase 2, if \(\EC_1'\) is in \(\Dd[d_1]\)
and \(\EC_2'\) is in \(\Dd[d_2]\),
then we need to compute a \((d_1,d_2)\)-crossroad
\(\EC_3'\) and find a path \(\EC_1'\to \EC_3'\) in \(\Dd[d_1]\)
and a path \(\EC_2' \to \EC_3'\) in \(\Dd[d_2]\).
(In particular, we should ensure that there exist supersingular
\((d_1,d_2)\)-crossroads before including \(d_1\) and \(d_2\) in \(D\).)

This is not worthwhile for large \(d\) or large \(D\).
Asymptotically, \(\#\Dd\)
is in \(O((\sum_{d\in D}\sqrt{d})\sqrt{p})\),
so the expected number of steps in Phase~1 
is reduced by a factor of \(O(\sum_{d\in D}\sqrt{d})\).
However, 
the individual steps become more expensive:
if we use modular polynomials to check membership of each~\(\Dd\),
then the number of \(\FF_{p^2}\)-operations per step grows linearly with
\(\sum_{d \in D}d\), overwhelming the benefit of the shorter walks.
Asymptotically, therefore, there is no benefit in taking large \(d\) or
large \(D\) in Phase~1.
(For more analysis of random walks into \((d,\pm1)\)-structures,
in different contexts,
see~\cite{2020/Eisentraeger--Hallgren--Leonardi--Morrison--Park}
and~\cite{2009/Charles--Goren--Lauter}.)

Generalized Delfs--Galbraith can become interesting
for \(D\) consisting of a few small \(d\),
however, precisely because the asymptotic 
\(\kappa(d,p) := \#\Dd/\#\Dd[1] \approx \sqrt{d}\)
no longer holds.
For \(d < 10\), for example,
we can have \(\kappa(d,p)\)
substantially greater than \(\sqrt{d}\)
(and also substantially less than 1).
For example,
if \(p\) is the toy SIDH-type prime
\(2^{52}\cdot3^{33}-1\),
then \(\kappa(5,p) \approx 4.916\).
If we can test for an isomorphism or \(5\)-isogeny to the conjugate
faster than we can compute six \(2\)-isogenies,
then we can take \(D = \{1,5\}\)
and walk into \(\Dd[1]\sqcup\Dd[5]\)
faster than walking into \(\Dd[1]\) alone.
This speedup is counterbalanced by a slowdown in Phase 2,
because walking in \(\Graph{\Dd[5]}\) costs more,
and because the walks there need to be a square-root of \(\kappa(5,p)\)
longer---though we can work modulo conjugation to mitigate this cost.

\subsection{\texorpdfstring{\lowercase{\((d,\epsilon)\)}}{(d,e)}-structures and SIDH graphs}

As we noted above,
the probability of a random walk in the supersingular \(\ell\)-isogeny graph 
hitting a vertex that is the image of a \((d,\epsilon)\)-structure
is very low.
It is even lower when we consider SIDH/SIKE graphs,
which cover only a very small proportion of the full isogeny graph,
resembling trees of walks of short, fixed length.

Nevertheless,
when we look at specific SIKE graphs,
we see that they contain sections of \(\ellGraph[2]{\Dd}\)
and \(\ellGraph[3]{\Dd}\) for various \(d\).
For example,
the starting curve in \texttt{SIKEp434} has a $d$-isogeny to its
conjugate for $d \in D = \{5, 13, 17, 29, 37, 41\}$ (and also for much
higher, but less practical values of~\(d\)).
If we consider the \(2\)-isogeny graph,
then we find that \(\ellGraph[2]{\Dd}\) passes through the starting curve
and continues down through the tree towards a public key
for \(d = 17\) and \(41\).
Hence, if we can find a \(2\)-isogeny path from a \texttt{SIKEp434} public key 
to a vertex in the image of \(\Dd[17]\) or \(\Dd[41]\),
then we have an express route to the starting curve.
Such an attack succeeds in a reasonable time
with only a very small probability,
but it is still devastatingly effective
for a tiny proportion of \texttt{SIKEp434} keys.

%----------------------------------------------------------------------------------------
%	 REFERENCES
%----------------------------------------------------------------------------------------

\printbibliography % Output the bibliography when using a .bib file

%----------------------------------------------------------------------------------------

\end{document}